\documentclass[11pt,reqno,twoside]{article}

\usepackage[hang]{footmisc}
\usepackage{lipsum}

\usepackage{fixltx2e} % Makes \( \) equation style robust, among other
\usepackage{cmap} % Load before fontenc 

\usepackage[T1]{fontenc}
\usepackage[utf8]{inputenc}
\usepackage{graphicx}
\usepackage{placeins}
\usepackage{enumerate}
\usepackage{algcompatible}
\usepackage{subcaption}
\usepackage{verbatim}
\newcommand{\comments}[1]{}

\usepackage{soul}

\usepackage{setspace}

\let\counterwithin\relax  %DSA: i had to include this to be able to compile
\usepackage{lmodern} % Improved version of computer modern
\usepackage[scale=0.88]{tgheros} % Helvetica clone for sans serif font

\usepackage{bm} % boldmath must be called after the package

\usepackage{bbold}

\usepackage{amsmath,amsbsy,amsgen,amscd,amsthm,amsfonts,amssymb} 

\usepackage[centering,top=1.2in,bottom=1.2in,left=0.8in,right=0.8in]{geometry}

\usepackage{titling}
\usepackage{musicography}
\graphicspath{ {./images/} }

\usepackage[sf,bf,compact]{titlesec}

\usepackage{booktabs,longtable,tabu} % Nice tables
\usepackage[font=small,margin=30pt,labelfont={sf,bf},labelsep={space}]{caption}

\usepackage[usenames,dvipsnames]{xcolor}

\definecolor{dark-gray}{gray}{0.3}
\definecolor{dkgray}{rgb}{.4,.4,.4}
\definecolor{dkblue}{rgb}{0,0,.5}
\definecolor{medblue}{rgb}{0,0,.75}
\definecolor{rust}{rgb}{0.5,0.1,0.1}

\usepackage{url}
\usepackage[colorlinks=true]{hyperref}
\hypersetup{linkcolor=dkblue}    
\hypersetup{citecolor=rust}      
\hypersetup{urlcolor=rust}     

\usepackage[final]{microtype} 

\newtheoremstyle{myThm} % name
    {\topsep}                    % Space above
    {\topsep}                    % Space below
    {\itshape}                   % Body font
    {}                           % Indent amount
    {\sffamily\bfseries}                   % Theorem head font
    {.}                          % Punctuation after theorem head
    {.5em}                       % Space after theorem head
    {}  % Theorem head spec (can be left empty, meaning ‘normal’)

\newtheoremstyle{myRem} % name
    {\topsep}                    % Space above
    {\topsep}                    % Space below
    {}                   % Body font
    {}                           % Indent amount
    {\sffamily}                   % Theorem head font
    {.}                          % Punctuation after theorem head
    {.5em}                       % Space after theorem head
    {}  % Theorem head spec (can be left empty, meaning ‘normal’)

\newtheoremstyle{myDef} % name
    {\topsep}                    % Space above
    {\topsep}                    % Space below
    {}                   % Body font
    {}                           % Indent amount
    {\sffamily\bfseries}                   % Theorem head font
    {.}                          % Punctuation after theorem head
    {.5em}                       % Space after theorem head
    {}  % Theorem head spec (can be left empty, meaning ‘normal’)

\theoremstyle{myThm}
\newtheorem{theorem}{Theorem}[section]

\newtheorem{proposition}[theorem]{Proposition}
\newtheorem{corollary}[theorem]{Corollary}

\theoremstyle{myRem}
\newtheorem{remark}[theorem]{Remark}

\theoremstyle{myDef}

\usepackage{fancyhdr}
\usepackage{algorithm}
\usepackage{algpseudocode}

\let\originalleft\left
\let\originalright\right
\renewcommand{\left}{\mathopen{}\mathclose\bgroup\originalleft}
\renewcommand{\right}{\aftergroup\egroup\originalright}

\usepackage{mathtools}
\mathtoolsset{centercolon}  % Makes := typeset correctly for definitions

\renewcommand{\phi}{\varphi}

\providecommand{\mathbbm}{\mathbb} % In case we don't load bbm

\newcommand{\R}{\mathbbm{R}}

\definecolor{mygreen}{rgb}{0.1,0.75,0.2}

\newcommand{\nc}{\normalcolor}

\newcommand{\Expect}{\operatorname{\mathbb{E}}}

\newcommand{\dkl}{d_{\mbox {\tiny{\rm KL}}}}

\newcommand{\Nc}{\mathcal{N}}

\newcommand{\qlp}{q_{\mbox {\tiny{\rm LP}}}}

\newcommand{\zlp}{z_{\mbox {\tiny{\rm LP}}}}

\newcommand{\Clp}{C_{\mbox {\tiny{\rm LP}}}}

\newcommand{\bt}{\theta}
\newcommand{\bti}{\theta_i}

\newcommand{\ELBO}{\textsc{elbo}}
\newcommand{\g}{\,\vert\,}

\newcommand{\EEq}[1]{\mathbb{E}_{q}\left[#1\right]}
\newcommand{\EEt}[1]{\mathbb{E}_{q(\theta)}\left[#1\right]}
\newcommand{\EEu}[1]{\mathbb{E}_{q(u)}\left[#1\right]}

\newcommand{\J}{{\mathsf{J}}}

\usepackage[font = small, margin=30pt]{caption}
\usepackage[]{algorithm}
\usepackage{enumerate}

\usepackage{graphicx}

\usepackage{authblk}
\usepackage[square,numbers]{natbib}
\usepackage{chngcntr}
\usepackage{mathrsfs}
\counterwithin{table}{section}
\counterwithin{algorithm}{section}

\title{A Variational Inference Approach to \\ Inverse Problems with Gamma Hyperpriors} 

\author{Shiv Agrawal,  Hwanwoo Kim, Daniel Sanz-Alonso, and Alexander Strang}

\date{University of Chicago}

\vspace{.25in}

\makeatletter\@addtoreset{section}{part}\makeatother%
\numberwithin{equation}{section}

\newcommand{\upperRomannumeral}[1]{\uppercase\expandafter{\romannumeral#1}}

\begin{document}
\maketitle %  LEAVE HERE
% The command above causes the title to be displayed.

%>>>>> DELETE ALL CONTENT UNTIL "\end{document}"
% This is the body of your document.

\abstract{Hierarchical models with gamma hyperpriors provide a flexible, sparse-promoting framework to bridge $L^1$ and $L^2$ regularizations in Bayesian formulations to inverse problems. Despite the Bayesian motivation for these models, existing methodologies are limited to \textit{maximum a posteriori} estimation. The potential to perform uncertainty quantification has not yet been realized. 
This paper introduces a variational iterative alternating scheme for hierarchical inverse problems with gamma hyperpriors. The proposed variational inference approach yields accurate reconstruction, provides meaningful uncertainty quantification, and is easy to implement. In addition, it lends itself naturally to conduct model selection for the choice of hyperparameters. We illustrate the performance of our methodology in several computed examples, including a deconvolution problem and sparse identification of dynamical systems from time series data.  

\section{Introduction}\label{sec:introduction}
This paper introduces a variational inference approach that enables uncertainty quantification for hierarchical Bayesian inverse problems with gamma hyperpriors. The hierarchical model that we consider, along with an Iterative Alternating Scheme (IAS) to compute the \textit{maximum a posteriori} (MAP) estimate, were introduced and analyzed in   \cite{calvetti2020sparse,calvetti2019hierachical,calvetti2020sparsity,calvetti2019brain,calvetti2015hierarchical}. These papers provide strong evidence of the flexibility of the hierarchical model and show that the IAS algorithm is easy to implement and globally convergent. However, despite the Bayesian motivation for the hierarchical model, previous work has only considered MAP estimation, and the potential to perform uncertainty quantification has not yet been realized. Using the general framework of variational inference, we introduce a Variational Iterative Alternating Scheme (VIAS)  that  shares the flexibility and ease of implementation of IAS, while enabling uncertainty quantification and model selection. 

The hierarchical Bayesian model that we consider gives a posterior density $p(u,\theta \g y)$ for the unknown quantity of interest $u \in \R^d$ and parameters $\theta \in \R^d$ given observed data $y \in \R^n.$ The goal of IAS is to find the MAP estimator, that is, the pair $(u^*, \theta^*)$ that maximizes the posterior density. This leads to an optimization problem which IAS solves by producing iterates $(u^k,\theta^k),$ $k \ge 1,$  satisfying
\begin{align}\label{eq:IASintro}
\begin{split}
u^{k+1} &= \arg \max_u p(u, \theta^k \g y), \\
\theta^{k+1} &= \arg \max_\theta p(u^{k+1}, \theta \g y).
\end{split}
\end{align}
In contrast, the goal of our proposed VIAS method is to find the density $q^*(u,\theta)$ that is closest to the posterior $p(u,\theta \g y)$ in Kullback-Leibler divergence, within the mean-field family of distributions of the form $q(u,\theta) = q(u)\, q(\theta)$. This leads to an optimization problem over densities which VIAS solves by producing iterates $q^k(u,\theta)=q^k(u) \, q^k(\theta),$ $k \ge 1,$  satisfying
\begin{align}\label{eq:VIASintro}
\begin{split}
q^{k+1}(\theta) &= \arg \min_{q(\theta)} \dkl \bigl(q^k(u) \, q(\theta) \, \| \, p(u,\theta \g y)\bigr), \\
q^{k+1}(u) &= \arg \min_{q(u)} \dkl \bigl(q(u) \, q^{k+1}(\theta) \, \| \, p(u,\theta \g y)\bigr).
\end{split}
\end{align}
Approximate Bayesian inference can then be performed using the variational distribution $q^*(u,\theta)$, which will be shown to be tractable, rather than the posterior $p(u,\theta \g y).$
Due to the tractability of $q^*(u,\theta),$ point estimates and credible intervals can be efficiently computed with the variational distribution, while doing so with the true posterior would be computationally challenging.

Central to the implementation of IAS is the fact that the maximizers $u^{k+1}$ and $\theta^{k+1}$ in \eqref{eq:IASintro} can be obtained in closed form, by exploiting the structure of the hierarchical model with gamma hyperpriors. A similar property is satisfied by VIAS. Indeed, our choice of mean-field admissible densities ensures that the minimizers $q^{k+1}(u)$ and $q^{k+1}(\theta)$ in \eqref{eq:VIASintro} are, respectively, Gaussian and generalized inverse Gaussian densities. We will derive closed formulas for the iterative updating of the parameters of these distributions.

Despite their shared structure, there are some fundamental differences between IAS and VIAS. While IAS only gives a point estimate i.e the MAP, VIAS gives a variational distribution that approximates the posterior. This variational distribution can be used to understand the covariance structure and find credible intervals for the estimates. However, it is worth emphasizing that VIAS only provides an approximation to the posterior, and therefore point estimates or credible intervals constructed with VIAS will only give approximate posterior inference. In contrast, IAS converges to the true MAP estimate. The primary advantage of VIAS is its potential to provide meaningful uncertainty quantification. We will also show that the variational perspective lends itself naturally to model selection for the choice of hyperparameters. 
An advantage of IAS is that it converges globally to the MAP estimate due to the convexity of the log-posterior density, while VIAS is, in general, only guaranteed to converge to a local maximizer of the optimization problem \eqref{eq:VIASintro}. We will demonstrate the potential emergence of spurious local  maxima in the VIAS objective function for extreme data realizations and hyperparameter values, and describe how convergence to the global  maximizer  can be achieved in practice by suitable initialization of the variational algorithm.

\subsection{Related Work}
This paper, among others, introduces variational inference techniques \cite{bishop,jordan1999introduction,wainwright2008graphical,blei2017variational}  to Bayesian inverse problems \cite{tarantola2005inverse,kaipio2006statistical,calvetti2007introduction,AS10,sanzstuarttaeb}, where computational approaches are often based on MAP estimation \cite{kaipio2006statistical}, Monte Carlo and measure transport sampling  \cite{liu2008monte,agapiou2017importance,marzouk2016sampling}, or iterative Kalman methods \cite{chada2020iterative}. Some recent works that have investigated the use of variational inference for inverse problems include \cite{maestrini2021variational,tonolini2020variational}. 
Variational inference has a comparable computational cost to MAP estimation, but has two main advantages: (i)  it can provide uncertainty quantification; and (ii) it lends itself naturally to conduct model selection. In addition, the variational distribution can be used as a proposal mechanism for Monte Carlo sampling algorithms. 
 A simple but popular alternate way to probe the posterior is to find its Laplace approximation, namely the Gaussian centered at the MAP whose covariance is given by the inverse Hessian of the negative log-posterior density. The Kullback-Leibler accuracy of Laplace approximations was investigated in \cite{dehaene2017computing}, and the Hellinger accuracy in inverse problems with small noise was established in \cite{schillings2020convergence}. However, Laplace approximations can be ineffective in large-noise or small-data regimes, where the posterior may not be well approximated by a Gaussian measure. In addition, computing the inverse Hessian can be prohibitively expensive in high dimensional nonlinear inverse problems.  Monte Carlo methods can provide accurate posterior inference while variational inference is based on an approximation to the posterior; however, Monte Carlo methods often require a large number of samples and hence a large number of forward model evaluations, which can be costly. In addition, tuning Monte Carlo methods and assessing their convergence can be challenging. For hierarchical Bayesian models, the Gibbs sampler alleviates the need of tuning \cite{damlen1999gibbs}, but  the chain may still converge slowly for highly anisotropic target densities \cite{agapiou2014analysis,roberts1997updating}. 

As mentioned above, we will consider a hierarchical Bayesian model with gamma hyperpriors introduced and analyzed in \cite{calvetti2020sparse,calvetti2019hierachical,calvetti2020sparsity,calvetti2019brain,calvetti2015hierarchical}. The paper \cite{calvetti2020sparse} investigates generalized gamma hyperpriors and \cite{calvetti2020sparsity} discusses hybrid solvers for MAP estimation that can improve on IAS. The hierarchical model and IAS algorithm have been shown to be successful in realistic inverse problems including brain activity mapping from MEG \cite{calvetti2015hierarchical,calvetti2019brain}. These papers also emphasize the flexibility of the model, and its ability to provide useful regularization for sparse signals \cite{calvetti2019hierachical,calvetti2020sparse,calvetti2020sparsity}. As described in \cite{calvetti2019hierachical},  the IAS algorithm is closely related to iterative reweighted least squares  \cite{green1984iteratively} and related work \cite{gorodnitsky1997sparse,daubechies2010iteratively} on signal processing with emphasis on sparsity. Sparse-promoting algorithms and models are key in statistics applications \cite{tibshirani1996regression,carvalho2009handling}. Our hierarchical approach is closely related to empirical Bayes statistical methods \cite{robbins1992empirical} and to bilevel and data-driven methods for inverse problems \cite{bard2013practical,arridge2019solving}.

\subsection{Outline and Main Contributions}
\begin{itemize}
\item Section \ref{sec: hierarchical} formalizes the problem of interest and reviews the hierarchical model with gamma hyperpriors and the IAS algorithm. Building on previous work on IAS \cite{calvetti2019hierachical}, we derive and show the convergence of an iterative Laplace approximation to the posterior, used in Section \ref{sec: experiments} for numerical comparisons with our proposed VIAS. 
\item Section \ref{sec:VI} introduces the novel VIAS and discusses its convergence. We will give all necessary background on variational inference. 
\item Section \ref{sec: experiments} demonstrates the accuracy of VIAS and its ability to provide meaningful uncertainty quantification in four computed examples. These examples include a deconvolution problem from \cite{calvetti2020sparse} and a new application of IAS and VIAS for data-driven sparse identification of dynamical systems \cite{brunton2016discovering} from time series data. We also introduce and show the effectiveness of a model selection approach for the choice of hyperparameters. 
\item We close in Section \ref{sec:Conclusions} with some research directions that stem from this work. 
\end{itemize}

\paragraph{Notation} For matrix $P,$ we write $P\succ 0$ if $P$ is positive definite. For $P \succ 0,$ we denote by $\| \cdot \|^2_P := | P^{-1/2} \cdot |^2$ the squared Mahalanobis norm induced by the matrix $P,$ where $| \cdot |$ denotes the Euclidean norm. 

\section{MAP Estimation and Laplace Approximation}\label{sec: hierarchical}
This section is organized as follows. In Subsection \ref{ssec:21}, we formalize the inverse problem of interest and the hierarchical model with gamma hyperpriors. Subsection \ref{ssec:22} overviews the IAS algorithm for MAP estimation and introduces an iterative Laplace approximation method.  Subsection \ref{ssec:23} reviews a convergence result for the IAS algorithm, from  which we deduce convergence of the iterative Laplace approximation method.

\subsection{Hierarchical Bayesian  Model}\label{ssec:21}
We consider the following linear discrete inverse problem of recovering an unknown $u$ from data $y$ related by
\begin{equation}\label{eq:IP}
 y = A u + \eta,
\end{equation} 
where $A \in \R^{n \times d}$ is a given, possibly ill-conditioned, matrix and typically $d \ge n.$ We assume that the noise term $\eta$ is Gaussian distributed $\eta \sim \Nc(0, \Gamma)$ with given $\Gamma \succ 0.$ Following \cite{calvetti2020sparse,calvetti2019hierachical,calvetti2020sparsity,calvetti2019brain,calvetti2015hierarchical}, we adopt the following hierarchical Bayesian model, where the prior on $u$ is conditionally Gaussian given a prior variance vector $\theta \in \R^d$:
\begin{align}\label{eq:model}
\begin{split}
y \g u &\sim \Nc(Au, \Gamma ), \\
u\g\theta &\sim  \Nc(0, D_\theta), \quad \quad \quad D_\theta = \text{diag}(\theta), \\
 \theta_i &\sim \text{Gamma}(\alpha_i, \beta), \,\, \,\, \, 1 \le i \le d.
 \end{split}
\end{align}
Here $\alpha_i$ and $\beta$ denote the shape and rate parameters, respectively. Our aim is to estimate  $z := (u, \theta)$ given the observed data $y.$ In the Bayesian approach to inverse problems \cite{kaipio2006statistical,AS10,sanzstuarttaeb}, inference is based on the posterior distribution which, for the hierarchical model \eqref{eq:model}, takes the form
\begin{align}\label{eq:posterior}
\begin{split}
p   ( u \, , \theta  \g y ) &= \frac{ p(y \g u,\theta)    p(u \g \theta)  p(\theta)}{p(y)}  \propto \exp \bigl( -\J(u,\theta) \bigr),
\end{split}
\end{align}
where 
\begin{equation}\label{eq:energy}
\ooalign{
$\J(u,\theta ) := \overbrace{ \frac{1}{2} \| y - Au \|^2_{ \Gamma } +  \frac{1}{2}\| u \|^2_{D_\theta}  }^{(a)} + \sum_{i = 1}^{d} \biggl[ \frac{\theta_i}{\alpha_i} - \Bigl( \beta - \frac32 \Bigr) \log \frac{\theta_i}{\alpha_i}  \biggr] .$ \cr 
$\phantom{\J(u,\theta ) = \frac{1}{2} \| y - Au \|^2  + {} } {\underbrace{ \phantom{  \frac{1}{2}\| u \|^2_{D_\theta} \sum_{j = 1}^d \biggl[ \frac{\theta_i}{\alpha_i} - \Bigl( \beta - \frac32 \Bigr) \log \frac{\theta_i}{\alpha_i}  \biggr] \,\,\,\,\,\,\,\,\,\,} }_{(b)} }$ \cr 
}
\end{equation}
Here $(a)$ and $(b)$ identify the two objectives that will be minimized iteratively by IAS. 
The $\alpha_i$'s act as scale parameters that control the expected size of $\theta_i$, and, as a result, $u_i^2$ \cite{calvetti2020sparse}. They can be chosen automatically using the signal to noise ratio and expected cardinality of the support \cite{calvetti2019hierachical}. Previous work \cite{calvetti2020sparse} has analyzed a whitening of the problem, setting all $\alpha_i = 1$. In contrast, in our variational algorithm we will not perform such whitening, and the $\alpha_i$'s will determine the degree of shrinkage towards zero off the support of the unknown. The hyperparameter $\beta$ controls how sharply $\J(u,\theta)$ penalizes non-sparse inputs for small but non-zero values off the support. In the limit as $\beta$ converges to $3/2,$ the MAP estimator given by the minimizer of $\J(u,\theta)$ converges to the solution to an $L^1$ penalized problem \cite{calvetti2019hierachical}. 
We refer to \cite{calvetti2015hierarchical} for further background and motivation on the use of hierarchical gamma hyperpriors in Bayesian inverse problems which, contrary to common practice in statistics, do not lead to a conjugate model.  %The broader family of generalized-gamma hyperpriors has been considered in \cite{calvetti2020sparse}.

\subsection{The Iterative Alternating Scheme and Laplace Approximation}\label{ssec:22}
The MAP estimate of $z=(u,\theta)$ is, by definition, the maximizer of the posterior $p(z\g y)$ or, equivalently, the minimizer of $\J(z).$ The papers \cite{calvetti2020sparse,calvetti2019hierachical,calvetti2020sparsity,calvetti2019brain,calvetti2015hierarchical} proposed, analyzed, and implemented an Iterative Alternating Scheme (IAS) for MAP estimation in a variety of inverse problems. The IAS consists of two separate minimization steps:
\begin{enumerate}
\item {\bf Initialize} $\theta^0,$ $k = 0$. 
\item {\bf Iterate  until convergence:}
\begin{enumerate}[(i)]
\item  Update $u^{k+1} = \arg \min_u  \J(u, \theta^k).$
\item Update $\theta^{k+1} = \arg \min_\theta \J(u^{k+1}, \theta).$ 
\item  $k \rightarrow k+1.$
\end{enumerate}
\end{enumerate}
Let $z^k := (u^k,\theta^k).$ Clearly $\J(z^k)$ is monotonically decreasing in $k.$ Under suitable assumptions on the hyperparameters, to be made precise in Proposition \ref{prop:IASconvergence} below, $\J$ is convex and IAS converges to the global minimizer of $\J.$ In other words, $z^k$ convereges to the MAP estimator. In addition to this convergence guarantee, the IAS algorithm is simple to implement because of the structure of the energy functional $\J.$ Indeed,  in step (i) only the $u$-dependent part $(a)$ in  \eqref{eq:energy} needs to be considered, and in step (ii) only the $\theta$-dependent part $(b)$ is needed. This results in straightforward implementation of both steps, as we describe next.

\subsubsection{Updating $u$}
The update of the $u$ component boils down to solving a standard linear least-squares problem, which admits a closed form solution
\begin{align}
\arg \min_u  \J(u, \theta) &= \arg \min_u  \frac{1}{2} \| y - Au \|^2_{ \Gamma } +  \frac{1}{2}\| u \|^2_{D_\theta} \label{eq:constrainedLinear} \\
& = (A^\top \Gamma^{-1} A + D_\theta^{-1} )^{-1} A^\top \Gamma^{-1} y. \nonumber
\end{align}
In practice, when the dimension $d$ of the unknown is much larger than the dimension $n$ of the data, this linear least-squares problem can be effectively solved using conjugate gradient together with an early stopping based on Morozov's discrepancy principle. This approach has been applied in \cite{calvetti2020sparse,calvetti2019hierachical,calvetti2020sparsity,calvetti2019brain,calvetti2015hierarchical}, and further analyzed in \cite{calvetti2018bayes}. In such underdetermined problems, inverting in $d$-dimensional space can also be avoided using the Sherman-Morrison-Woodbury lemma, which gives the following equivalent Kalman-type update 
\begin{equation}\label{eq:KalmanIAS}
\arg \min_u  \J(u, \theta) = Gy,  \quad  \quad G := D_\theta A^\top(AD_\theta A^\top + \Gamma)^{-1},
\end{equation}
where the matrix $G$ is called the Kalman gain. 

Note that \eqref{eq:constrainedLinear} can be rewritten as
$$
\arg \min_u  \J(u, \theta) = \arg \min_u  \frac{1}{2} \| y - Au \|^2_{ \Gamma } +  \frac{1}{2} \sum_{i=1}^d \frac{u_i^2}{\theta_i},
$$
which shows that $D_\theta$ controls the sparsity of the solution $u$, with smaller $\theta_i$ leading to more shrinkage of $u_i$ towards zero. Therefore, in the hierarchical Bayesian model setup, the variance parameter $\theta$ not only determines the variation of the parameter $u$ but also the level of sparsity of $u$.

\subsubsection{Updating $\theta$}
As shown in \cite{calvetti2019hierachical}, the update of the $\theta$-component part (b) in \eqref{eq:energy} can be obtained by direct computation of a critical point of $\J(z)$ as follows
\begin{align*}
\arg \min_\theta \J(u, \theta) =
 \alpha_i \biggl( \frac{\tilde{\beta} }{2} + \sqrt{ \frac{\tilde{\beta}^2}{4} + \frac{u_i ^2}{2\alpha_i} } \,\, \biggr), \quad \tilde{\beta} = \beta - 3/2.
\end{align*}
A pseudo-code for the IAS algorithm is given in Algorithm \ref{alg_1}. 

\begin{figure}
\makebox[\linewidth]{%
  \begin{minipage}{\dimexpr 0.75\linewidth+5em}
\begin{algorithm}[H]
\FloatBarrier
\caption{Iterative Alternating Scheme (IAS) \label{alg_1}}
\STATE {\bf Input}: Data $y$, matrix $A$. Prior hyperparameters: $ \alpha, \beta.$ \\
\STATE  Initialize $\theta^0,$ $k = 0$. \\
\STATE {\bf For} $k = 0, 1, \ldots$ until convergence {\bf do}:
\begin{enumerate}[(i)]
\item Set $D_\theta = \text{diag}( \theta^k)$ and update
  $$u^{k+1} = (A^\top \Gamma^{-1} A + D_\theta^{-1} )^{-1} A^\top \Gamma^{-1} y.$$ 
\item Update $$\theta_i^{k+1} =
 \alpha_i \biggl( \frac{\tilde{\beta} }{2} + \sqrt{ \frac{\tilde{\beta}^2}{4} + \frac{(u_i^{k+1} )^2}{2\alpha_i} } \,\, \biggr), \quad \tilde{\beta} = \beta - 3/2.$$
 \end{enumerate}
\STATE{\bf end for} \\
\STATE {\bf Output}: Approximation to the MAP estimator $(u^{k+1}, \theta^{k+1}) \approx \arg \max p(z\g y).$
\end{algorithm}
  \end{minipage}}
\end{figure}

\begin{remark}
A variety of stopping rules have been considered. For instance, the relative change in $u$ (or $u$ and $\theta)$ being below some threshold. As an alternative, the decrease in the two terms $(a)$ and $(b)$ in \eqref{eq:energy} can be monitored to direct stopping. We note again that the $u$ update in step (i) can be implemented using conjugate gradient together with a stopping criteria given by Morozov's discrepancy principle. 
\end{remark}

\subsubsection{IAS Laplace Approximation}
Here we show that the IAS iterates can be used to obtain a Laplace approximation to the posterior. Recall that the Laplace approximation $\qlp(z)  = \Nc(\zlp, \Clp)$ to the posterior $p(z \g y)$ is the Gaussian distribution whose mean $\zlp$ is the MAP estimator and whose precision $\Clp^{-1}$ is the Hessian of the objective $\J(z)$ evaluated at $\zlp, $ that is,
\begin{equation}\label{eq:LA}
\zlp = \arg \min_z \J(z) , \quad \quad \Clp^{-1} = \nabla \nabla \J(\zlp).
\end{equation}
Thus, the sequence $z^k = (u^k, \theta^k)$ can be used to approximate $\qlp(z)$ by the Gaussian $\qlp^k (z)= \Nc(\zlp^k, \Clp^k),$ 
where
\begin{equation}\label{eq:IASLA}
\zlp^k = (u^k, \theta^k), \quad  \quad \Clp^k = \Bigl( \nabla \nabla \J(\zlp^k) \Bigr)^{-1}.
\end{equation}
Partitioning the Hessian $H(z) = \nabla \nabla \J(z)$ into four blocks of size $d \times d$ gives
\begin{align*}
H(z) = 
\begin{bmatrix}
H_{uu}(z) & H_{u \theta}(z) \\
H_{\theta u}(z)  & H_{\theta \theta}(z)
\end{bmatrix},
\end{align*}
with
\begin{align*}
H_{uu} (z) &= A^\top  \Gamma^{-1} A + \text{diag} (1/\theta),\\
H_{u \theta}(z) &= - \text{diag} (u /\theta^2), \\ 
H_{\theta \theta}(z) &= \text{diag} (u^2/\theta^3 + \tilde{\beta} /\theta^2),
\end{align*}
where multiplication and division operations are defined in element-wise. This explicit characterization of the Hessian, together with Algorithm \ref{alg_1} and \eqref{eq:IASLA} yield an iterative Laplace approximation method.

\subsection{Convergence of IAS and Laplace Approximation}\label{ssec:23}
The following result was proved in \cite{calvetti2015hierarchical}. 
\begin{proposition}\label{prop:IASconvergence}
For $\beta > 3/2$ and $\alpha \in \R_+^n,$ the energy functional \eqref{eq:energy} defined over $\R^d \times \R_+^d$ is strictly convex, thus has a unique global minimizer $z^*= ( u^*, \theta^*).$ The IAS algorithm produces a sequence $z^k = (u^k, \theta^k)$ that converges to the global minimizer. 
\end{proposition}
The convergence analysis of IAS was further developed in \cite{calvetti2019hierachical}, where rates of convergence were established. As a consequence of Proposition \ref{prop:IASconvergence} we have the following corollary:
\begin{corollary}\label{corollary}
For $\beta > 3/2$ the IAS Laplace approximation $\qlp^k(z) = \Nc (z^k, \Clp^k)$ given by \eqref{eq:IASLA} converges weakly to the Laplace approximation $\qlp(z) = \Nc (\zlp, \Clp)$ given by \eqref{eq:LA}.
\end{corollary} 
\begin{proof}
Weak convergence of Gaussians is equivalent to convergence of their means and covariances \cite{bogachev1998gaussian}. The result follows from  Proposition  \ref{prop:IASconvergence}  and continuity of the Hessian. 
\end{proof}

\section{Variational Inference}\label{sec:VI}
In this section, we introduce our variational approach for posterior approximation. We provide the necessary background on variational inference in Subsection \ref{ssec:31}. The main algorithm is described in Subsection \ref{ssec:32}, and convergence guarantees are discussed in Subsection \ref{ssec:33}. Our presentation is parallel to that of the previous section. 

\subsection{Background and Mean-field Assumption}\label{ssec:31}
Variational inference is a popular technique \cite{bishop,jordan1999introduction,wainwright2008graphical,blei2017variational} for approximating the posterior distribution $p(z \g y)$ of some unknown parameter $z$ given  data $y$. We will be concerned with approximating the posterior $p(z \g y)$ given by \eqref{eq:posterior} with $z = (u,\theta).$ The goal of variational inference is to find an approximating distribution $q^*(z)$ which is close to the posterior, but tractable. Then, approximate Bayesian inference can be performed using $q^*(z)$ rather than $p(z\g y).$ The approximating distribution $q^*(z)$ is defined as the (numerical) solution to an optimization problem. Precisely, one specifies a family $\mathcal{D}$ of tractable distributions and sets 
\begin{equation}\label{eq:vi}
q^*(z) := \arg \min_{q \in \mathcal{D}} \, \dkl \bigl( q(z) \| p(z \g y) \bigr), 
\end{equation}
where $\dkl$ denotes the Kullback-Leibler divergence. The above minimization can be reformulated as maximizing  the \textit{evidence lower-bound}  (ELBO) given by:
\begin{align}
  \ELBO(q) &:= \EEq{\log p(z, y)} -  \EEq{\log q(z)}  \label{eq:elbo1} \\
   &= \log p(y) - \dkl \bigr(q(z) \| p(z \g y) \bigr) \label{eq:elbo2} \\
   & =  \EEq{\log p(y \g z)} -  \dkl \bigl(q(z) \| p(z) \bigr). \label{eq:elbo3}
\end{align}
Note from Equation \eqref{eq:elbo1} that $\ELBO(q)$ can be evaluated without computing the evidence $p(y),$ which is often intractable. Since the Kullback-Leibler divergence is non-negative, Equation \eqref{eq:elbo2} shows that $\ELBO(q)$ indeed provides a lower-bound on the log-evidence.  This property can be used for model selection, since larger ELBO indicates a higher probability of the data being generated by a particular model. Finally, Equation \eqref{eq:elbo3} shows that the optimal $q^*(z)$ finds a compromise between maximizing the expected log-likelihood and minimizing the Kullback-Leibler divergence to the prior $p(z)$.

For reasons discussed below, we will choose the variational family to be 
\begin{equation}
\mathcal{D} : = \Bigl\{ q(z) :  q(u, \bt) = q(u) q(\theta), \quad q(\theta) = \prod_{i=1}^{d} q(\bti) \Bigr\}.
\end{equation}
This \textit{mean-field} family is a popular choice in variational inference because it enables efficient numerical optimization of the ELBO using the Coordinate Ascent Variational Inference (CAVI) algorithm \cite{bishop}. Note, however, that under the mean-field approximation the variational distribution is unable to capture the dependence structure between $u$ and $\theta.$ This is not an assumption on the data model, but rather is implied by the choice of the variational family $\mathcal{D}$. In the next subsection, we derive a CAVI algorithm for the hierarchical Bayesian model \eqref{eq:model}. We shall see that this variational algorithm shares the ease of implementation of the IAS algorithm.

\subsection{The Variational Iterative Alternating Scheme (VIAS)}\label{ssec:32}
The variational distribution $q^*(z)$ is, by definition, the closest distribution in $\mathcal{D}$ to the posterior $p(z\g y)$, where closeness is quantified using the Kullback-Leibler divergence. Equivalently, $q^*(z)$ is the distribution that maximizes the ELBO in $\mathcal{D}.$ Here we propose and analyze a Variational Iterative Alternating Scheme (VIAS) to maximize the ELBO, consisting of the following two separate maximization steps:
\begin{enumerate}
\item {\bf Initialize} $q^0(u),$  $k = 0$. 
\item {\bf Iterate  until convergence:}
\begin{enumerate}[(i)]
\item Update $q^{k+1}(\theta) = \arg \max_{q(\theta)}  \ELBO \bigl( q^k(u)  q(\theta) \bigr) . $
\item Update $q^{k+1} (u) = \arg \max_{q(u)}   \ELBO \bigl( q(u)  q^{k+1}(\theta) \bigr)   .$
\item  $k \rightarrow k+1.$
\end{enumerate}
\end{enumerate}
Note that the structure of VIAS is identical to that of IAS, replacing the energy $\J(z)$ over unknown $u$ and parameters $\theta$ with the energy $\ELBO(q(z))$ over their joint distribution. Let $q^k(z) := q^k(u) q^k(\theta).$ By construction, $\ELBO(q^k(z))$ is monotonically increasing with $k$. In other words, the Kullback-Leibler divergence between $q^k(z)$ and the posterior decreases monotonically.

VIAS also shares with IAS its ease of implementation. The following well-known result \cite{bishop} gives a characterization for the maximizing distributions in steps (i) and (ii).
\begin{proposition}[Optimization of ELBO in Mean-field Variational Inference]\label{prop:ELBOmeanfield}
It holds that
\begin{align}
\arg \max_{q(u)}  \ELBO \bigl( q(u)  q(\theta) \bigr) \propto \exp \Bigl( \EEt{  \log p(y, z) } \Bigr),  \label{eq:quupdate}  \\
  \arg \max_{q(\theta)}  \ELBO \bigl( q(u)  q(\theta) \bigr)  \propto \exp \Bigl( \EEu{  \log p(y, z) } \Bigr). \label{eq:qtupdate}
  \end{align}
\end{proposition}
We next describe how these characterizations, which are a consequence of the mean-field assumption, imply that the maximizing distributions $q(u)$ and $q(\theta)$ in steps (i) and (ii) belong to certain parametric families. Precisely, we shall see in Subsection \ref{sssec:updatingqu} that \eqref{eq:quupdate} implies that $q(u) = \Nc(m,C)$ and in Subsection \ref{sssec:updatingqt}  that \eqref{eq:qtupdate} implies that $q(\theta_i) = \text{GIG}(b, r_i,s),$ where $\text{GIG}$ denotes the generalized inverse Gaussian distribution. These considerations will reduce the implementation of steps (i) and (ii) to an explicit recursion in the variational parameters. 

Before delving into the derivations, we recall for convenience, and later reference, that a random variable $\theta_i \sim \text{GIG}(b, r_i,s)$ has probability density function
\begin{equation}\label{eq:GIGpdf}
q(\bti \g b, r_i, s) = \frac{(b / r_i )^{s/2}}{2 K_s(\sqrt{ r_i b_i})} \bti^{s-1} e^{-(b \bti +  r_i/\bti)/2},
\end{equation}
where $K_s$ denotes the modified Bessel function of the second kind. Moreover, the following identities hold
\begin{align}\label{eq:identitiesGIG}
\begin{split}
\EEt{\bti}  & = \frac{ K_{s +1}(\sqrt{ r_i b}) }{ K_{s}(\sqrt{ r_i b})}  \cdot \sqrt{ r_i /b} \, ,\;  \\
\mathbb{V}_{q(\theta)} [ \bti ]  &= \frac{ K_{s +2}(\sqrt{ r_i b}) }{ K_{s}(\sqrt{ r_i b})}  \cdot ( r_i /b) - ( \EEt{\bti}) ^2  ,   \;    \\
\EEt{1/\bti} &=  \frac{ K_{s -1}(\sqrt{ r_i b}) }{ K_{s}(\sqrt{ r_i b})}  \cdot \sqrt{b/ r_i} \; .
\end{split}
\end{align}
The first and second identities can be used to compute credible intervals with the variational distribution, while the third identity will be used to derive the update for $q(\theta).$ For further properties of the generalized inverse Gaussian distribution, we refer to \cite{lemonte2011exponentiated}.

\subsubsection{Updating $q(u)$} \label{sssec:updatingqu}

To derive the update for $q(u)$, we use \eqref{eq:quupdate}. Note that
\begin{align*}
   \log q(u)  &\propto \EEt{\log p(y, u, \theta)}  \\
                &\propto \EEt{\log p(y \g u) + \log p(u \g \bt) + \log p(\bt) } \\
 & \propto  \EEt{-\frac12 |Au-y|^2_{ \Gamma } - \frac12 \sum_i \frac{u_i^2}{ \bt_i}} \\
% & \propto  - \frac12 |Au-y|^2 - \frac12 \sum_i u_i^2 \EEt{1/\bti} \\
  & \propto   -\frac12 |Au-y|^2_{\Gamma} - \frac12 u^\top L u,
\end{align*}
where $L = \text{diag}\Bigl( \EEt{1/ \theta} \Bigr)$. This implies that $q(u)$ is Gaussian with mean $m$ and covariance $C$ given by
\begin{align*}
m &= (A^\top \Gamma^{-1} A + L)^{-1} A^\top \Gamma^{-1} y, \\
C &= (A^\top \Gamma^{-1} A + L)^{-1}.
\end{align*}
The expectations $\EEt{1/ \theta_i}$ in the diagonal of $L$ can be obtained analytically using the fact (derived in the next subsection) that $q(\theta_i) = \text{GIG}(b, r_i,s), $ together with the third identity in \eqref{eq:identitiesGIG}.  Similar to IAS, the sparsity of the VIAS estimate of the parameter $u$ is controlled by the regularization coefficient matrix $L$. The larger the diagonal component of $L$ is, the smaller the corresponding component of $m$ will be.

\subsubsection{Updating $q(\theta)$} \label{sssec:updatingqt}
To update $q(\theta) = \prod_{i=1}^d q(\theta_i)$, we use independence and \eqref{eq:qtupdate} to obtain updates for each $q(\theta_i).$ Note that
\begin{align*}
\log q(\bti) &\propto \EEu{\log p(y \g u) + \log p(u \g \bt) + \log p(\bt) } \\
&\propto  \log p(\theta_i) +  \EEu{-\log \sqrt{\bti} - \frac{ u_i^2}{2\bti}} \\
&\propto \Bigl(\alpha- \frac32 \Bigr) \log \bti - \beta \bti - \frac{1}{2\bti}  \left(m_i^2 + C_{ii}\right),
\end{align*}
where in the last equality we used that $\EEu{u_i^2} = m_i^2 + C_{ii}$ for $q(u) = \Nc(m,C).$
Recalling \eqref{eq:GIGpdf}, this implies that $q(\theta_i) = \text{GIG}(b, r_i,s)$ where
$b = 2\beta$, $s = \alpha-0.5$ and $ r_i = m_i^2 + C_{ii}. $ 

Together, the update rule for $q(u)$ given $q(\theta)$, and $q(\theta)$ given $q(u)$, specify the VIAS. A pseudo-code for VIAS is given in Algorithm \ref{alg_2}.

\begin{figure}
\makebox[\linewidth]{%
  \begin{minipage}{\dimexpr 0.75\linewidth+5em}
\begin{algorithm}[H]
\caption{Variational Iterative Alternating Scheme (VIAS) \label{alg_2}}
\STATE {\bf Input}: Data $y$, matrix $A$. Prior hyperparameters: $ \alpha, \beta.$ \\
\STATE Initialize: $m^0, C^0, k=0.$ Set $b = 2 \beta$, $s = \alpha - 0.5.$ \\
\STATE {\bf For} $k = 0, 1, \ldots$ until convergence {\bf do}:
\begin{enumerate}[(i)]
\item Update $ r^{k+1}_i = (m^{k}_i)^2 + C^{k}_{ii} $
for each $i=1,\dots, d.$
\item Set 
\begin{equation}\label{eq:diagonal}
L = \text{diag}(\ell), \quad \quad \ell_i = \frac{ K_{s -1}\left(\sqrt{ r_i^{k+1}b}\right) }{ K_{s}\left(\sqrt{ r_i^{k+1}b}\right)}  \cdot \sqrt{\frac{b}{ r_i^{k+1}}}, \
\end{equation}
and update
\begin{align*}
m^{k+1} &= (A^\top \Gamma^{-1} A + L)^{-1} A^\top \Gamma^{-1} y , \\
C^{k+1} &= (A^\top \Gamma^{-1} A + L)^{-1}. 
\end{align*}
\end{enumerate}
\STATE{\bf end for} \\
\STATE{\bf Output}:  Variational approximation $p(z\g y) \approx q^{k+1}(z) = q^{k+1}(u) q^{k+1}(\theta),$ where
$$q^{k+1}(u) = \Nc(m^{k+1},C^{k+1}),  \quad q^{k+1}(\theta_i) = \text{GIG}(b, r_i^{k+1},s).$$
\end{algorithm}
  \end{minipage}}
\end{figure}

\begin{remark}
Using the Woodbury matrix identity, the Kalman-type expression in (\ref{eq:KalmanIAS}) can also be used to obtain VIAS updates for $m$ and $C$ without computing high dimensional matrix inversions 
\begin{align*}
m^{k+1} &= G y,  \quad \quad G:= L^{-1} A^\top(A L^{-1} A^\top +  \Gamma)^{-1},\\
C^{k+1}  &= (I-G A) L^{-1}.  
\end{align*}
The update for $m$ could also be implemented using conjugate gradient for least squares and an early stopping condition. 
\end{remark}

\begin{remark}
The ELBO can be computed at each iteration, and the relative change in the ELBO can be used as a stopping criteria, since this algorithm maximizes the ELBO. The relative change in the variational parameters along VIAS iterates could also be monitored to determine stopping. \qed
\end{remark}

\subsubsection{Variational Parameters and VIAS} 
On deriving the CAVI updates, we obtain that $q(u) = \Nc(m,C)$ and $q(\theta) = \text{GIG} (b,  r_i, s).$ The parameters $m, C, b,  r,$ and  $s$ are known as the variational parameters. 

 The parameters $m$ and $C$ are of the greatest interest to us, since they will determine the prediction of the unknown quantity of interest $u$. While $m$ gives the approximate posterior mean,  $C$ will allow us to obtain credible intervals on the prediction, and to understand the correlation between different components of $u$. 
 
 In the distribution of $\theta$, the values of the parameters $b$ and $s$  are directly related to the hyperparameters describing the prior gamma distribution: $b = 2\beta, s = \alpha-0.5$.
 Thus, the choice of hyperparameters directly affects the variational posterior. 
 Note that each diagonal component of the matrix $L$ in \eqref{eq:diagonal} will be large if $b$ or $\beta$ are large. On the other hand, each diagonal element is a decreasing function in $s \in (-0.5, 0.5) $ and therefore in $\alpha \in (0, 1)$. 
As a consequence, if one expects sparse structure in the true parameter $u$, choosing a small $\alpha$ value with a moderately large $\beta$ value would lead to adequate shrinkage. Each diagonal component of the matrix $L$ diverges to infinity as $\beta$ increases. So, to avoid shrinking each component of the parameter estimate too close to zero, one should not use extremely large $\beta$ value.  The mean of the prior $\alpha/\beta$ should be chosen to be close to the expected size of the unknown $\theta$, if prior information on this size is available. In addition to these heuristics, we will illustrate in Section \ref{sec: experiments} how the hyperparameters can be learned by a simple model selection procedure. Our numerical experiments show that the reconstructions are not sensitive to perturbation of the model hyperparameters, but that obtaining appropriate ballpark values for the hyperparameters through model selection can substantially improve the reconstruction. 

%To compare the role of hyperparameters $\alpha$ and $\beta$ in IAS and VIAS algorithms, we have provided the following table to clarify the effect of each parameter.

%\begin{center}
%\begin{tabular}{c||cc} 
%  \hline
%   &  $\alpha$ & $\beta$ \\ 
%  \hline
%    \hline
%  IAS & Controls the expected size of $\theta$ & Controls the sparse level of $u$ \\ 
%  \hline
%  VIAS &  Controls the sparse level of $u$ &  Controls the sparse level of $u$ \\ 
%  \hline
%\end{tabular}
%\end{center}

All the information about the variational distribution is stored in the five parameters; two of them are fixed, and the other three are interdependent. The CAVI algorithm updates these three parameters iteratively as follows: 
\begin{itemize}
\item Keeping $m$ and $C$ constant,  update each $ r_i$ with the formula:
$$ r_i = m_i^2 + C_{ii}.$$
\item Keeping $ r $ constant, update $m$ and $C$:
\begin{align*}
C &= (A^\top \Gamma^{-1} A + L)^{-1} \text{   with  } L_{ii} = \Expect \Bigl[\frac{1}{\theta_i}\Bigr], \\
 m &= (A^\top \Gamma^{-1} A + L)^{-1} A^\top \Gamma^{-1} y. 
\end{align*}
\end{itemize}

\subsection{Convergence and Initialization of VIAS}\label{ssec:33}
The $\ELBO$ in \eqref{eq:quupdate} and \eqref{eq:qtupdate} can now be written as a function of the variational parameters rather than the variational distribution, i.e., $\ELBO \bigl( q(u)  q(\theta)\bigl) = \ELBO(m,C,r)$. Then, the CAVI updates can be rewritten  
\begin{align*}
     \arg \max_{q(\theta)}\ELBO \bigl( q(u)  q(\theta)\bigl) &= \arg \max_{r}\ELBO \bigl(m, C,  r\bigl), \\
     \arg \max_{q(u)}\ELBO \bigl( q(u)  q(\theta)\bigl) &= \arg \max_{m, C}\ELBO \bigl(m, C,  r\bigl).
\end{align*}
As the parameters updated through VIAS are $m, C$ and $r$, we will ignore terms in the ELBO that do not depend on them. We will still denote the remaining expression as $\ELBO(m,C,r)$. A straightforward calculation shows that
\begin{equation*} %\label{eq:ELBO}
 \ELBO(m,C,r)  = -\frac{1}{2}\text{tr}(ACA^\top) - \frac{1}{2}|Am-y|^2 + \frac{1}{2}\log\text{det} C
- \frac{s}{2}\sum_{i=1}^d \log \frac{b}{r_i} + \sum_{i=1}^d \log 2 K_s(\sqrt{r_ib}).
\end{equation*}
The following result, which follows from Theorem 2.2 of \cite{bezdek1987local}, shows local convergence of VIAS.
\begin{proposition}\label{prop:VIASconvergence}
Suppose that $\ELBO(m,C,r)$ has a local maximum at $(m^*, C^*, r^*)$ and that the Hessian of $\ELBO(m,C,r)$ is negative definite at $(m^*, C^*, r^*)$. Then there is a neighborhood $\mathcal{U}$ of $(m^*, C^*, r^*)$ such that, for any initialization $(m^{0} , C^{0}, r^{0}) \in \mathcal{U}$, VIAS converges to $(m^*, C^*, r^*)$.
\end{proposition}
Unfortunately, the convergence to the global maximum of $\ELBO(m,C,r)$ is not guaranteed. This is because the function $\ELBO(m,C,r)$ can have multiple local maxima. We illustrate this phenomenon in the univariate case $u\in \R$, where $y = Au + \eta, \, \eta \sim \Nc(0,I_n).$ Denoting $C = c \in (0,\infty)$ and setting $b = 1,$ we then have
\begin{equation}\label{eq:1DELBO}
\ELBO(m,c,r) = -\frac{A^\top A}{2}(m^2 + c) + m^\top A^\top y + \frac{1}{2}\log c + \frac{s}{2}\log r + \log\left(2K_s(\sqrt{r}) \right).
\end{equation}
VIAS maximizes \eqref{eq:1DELBO} along the manifold $\mathcal{M}$ given by
$$
\mathcal{M} = \{(m,c,r)|~m = c A^\top y, ~r = (A^\top y)^2c^2 + c \}.
$$
Denoting $y_A := A^\top y$, the expression \eqref{eq:1DELBO} on this manifold becomes
\begin{equation}\label{eq:ELBOmanifold}
-\frac{A^\top A}{2}(c + y_A^2c^2) + y_A^2 c + \frac{1}{2}\log c + \frac{s}{2}\log \left(y_A^2c^2 + c\right) + \log\left(2K_s\left(\sqrt{y_A^2c^2 + c} \right) \right).
\end{equation}

To gain a better understanding of the $\ELBO(m,c,r)$ on the manifold $\mathcal{M}$, we provide plots of \eqref{eq:ELBOmanifold} with $n=50$ for the following three cases: 1) $A^\top A  = 1, y_A = 3, s = -0.49$; 2) $A^\top A = 2, y_A = 4, s = -0.499;$ and 3) $A^\top A = 5, y_A = 10, s = -0.499$. Figure \ref{Figure0} shows that there are multiple local maxima, which suggests that the initialization of VIAS can have an effect on the final point of convergence. The plots in Figure \ref{Figure0} indicate that, in each case, the global maximizer is the local maximizer farthest away from zero. For this reason, we recommend initializing VIAS with a covariance $C^{0} \succ \lambda I_d$, with $\lambda>0$ far away from zero. In addition, we expect $|y_A|$ to be large when the noise level is high. In such a case the global maximum of \eqref{eq:ELBOmanifold} will be far from the origin, as seen in Figure \ref{Figure1}, which further justifies the suggested initialization. From the perspective of quantifying uncertainties in $u$, initializing at a large covariance ensures convergence of the VIAS iterates $C^k$ to a matrix $C$ that gives conservative credible intervals for the reconstruction.

\begin{figure}[H]
		\centering 
	\includegraphics[height = 4.5cm, width = 5cm]{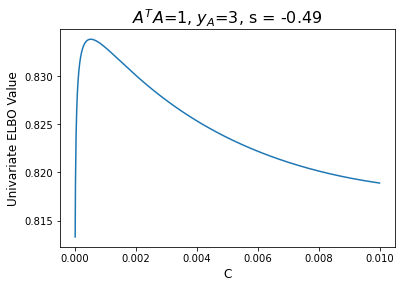}
	\includegraphics[height = 4.5cm, width = 5cm]{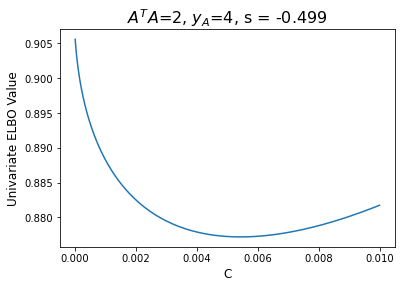}
	\includegraphics[height = 4.5cm, width = 5cm]{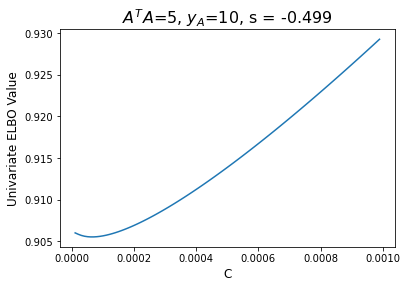}
	
	\includegraphics[height = 4.5cm, width = 5cm]{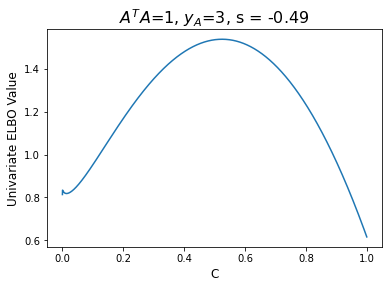}
	\includegraphics[height = 4.5cm, width = 5cm]{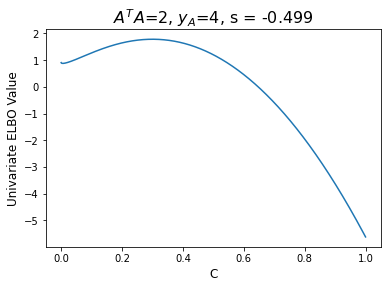}
	\includegraphics[height = 4.5cm, width = 5cm]{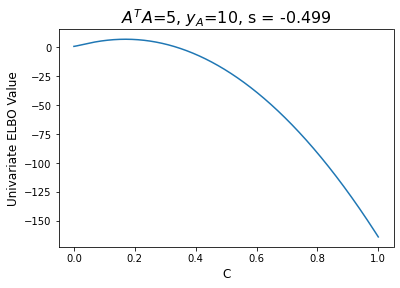}
	\caption{\label{Figure0} Top row: plots of \eqref{eq:ELBOmanifold} near zero. Bottom row: plots of \eqref{eq:ELBOmanifold} for [0,1].}
\end{figure}

\begin{figure}[H]
		\centering 
	\includegraphics[height = 4.5cm, width = 5cm]{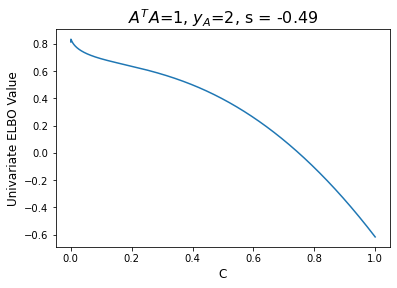}
	\includegraphics[height = 4.5cm, width = 5cm]{ELBO_ATA1_YTA3_S049_1.png}
	\includegraphics[height = 4.5cm, width = 5cm]{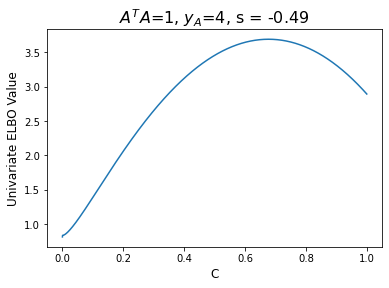}
	\caption{\label{Figure1} Plots of \eqref{eq:ELBOmanifold} for $y_A = 2, 3, 4$ with $A^\top A = 1, s = -0.49.$ }
\end{figure}

In Section \ref{sec: experiments}, we introduce hyperparameter tuning based on the ELBO values. Typically, the calibrated $s$ values were near -0.5. Accordingly, we have characterized a region of $(|y_A|, A^\top A)$ values where the ELBO  has more than one local maximum for $s = -0.499$. We used a grid-search to find a local maximum with a mesh step size of $10^{-7}$ ranging from zero to one. Multiple local maxima occurred in the beige-colored region in Figure \ref{Figure01}. In the case where multiple local maxima exist, we observed that the global maximum is the farthest away from zero. For smaller $s$ values,  i.e., closer to -0.5, which would promote sparse structure in the VIAS estimate, we observed a similar pattern to the one in the left plot of Figure \ref{Figure01}. 

To assess if data realizations that give multiple local maxima are likely to occur, we  ran an empirical study to estimate $\mathbb{P}(|y_A| \ge k)$ for $k \in \{0, \ldots, 10\}$. To do so, we first sampled $\theta$ from the gamma distribution with shape parameter 0.001 and rate parameter 1. In addition, each component of the vector $A$ satisfying the prespecified $A^\top A$ value was obtained from a uniform distribution. Next we generated a scalar $u$ from the Gaussian distribution with mean zero and variance $\theta$. Then we randomly sampled $y = Au + \eta$, where $\eta \sim N(0, I_n)$ for $10^4$ times and obtained the proportion of times when the event $\{|y_A| \ge k\}$ occurs. From Figure \ref{Figure01}, we can observe that such an event can occur with a zero probability.

In all our experiments, the global maximum of the ELBO was the local maximum farthest away from zero.  Based on such experiments and on the computed examples in the next section, we believe VIAS is very likely to converge to the global maximum of the ELBO in most practical applications, as long as it is initialized as suggested above.

 \begin{figure}[H]
		\centering 
	\includegraphics[height = 4.5cm, width = 5.5cm]{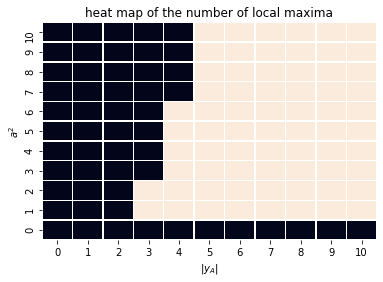}
	\includegraphics[height = 4.5cm, width = 6.5cm]{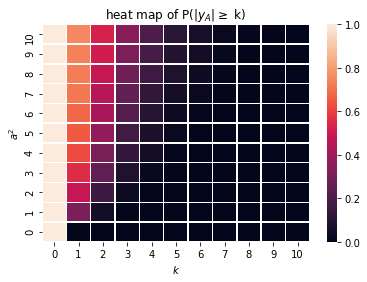}
	\caption{\label{Figure01}
	Common parameter: $s = -0.499$.
	Left: heat map of the number of local maxima of the ELBO (black: one maximum, beige: more than one maxima). Right: heat map of the approximated $\mathbb{P}(|y_A| \ge k)$. Together, the two plots show that it is unlikely to observe data that gives an ELBO with more than two local maxima.} 
\end{figure}

\section{Computed Examples} \label{sec: experiments}
In this section, we report the performance of VIAS in four computed examples, assessing its accuracy and its ability to provide meaningful uncertainty quantification. We also explore how to exploit the variational inference framework to guide the choice of model hyperparameters. 

\subsection{Truth and Data from Hierarchical Model}\label{ssec:Hiearchical}
We first apply VIAS to data generated from the hierarchical model  \eqref{eq:model}. This serves to illustrate the role of the hyperparameters in the hierarchical model, and also the application of our proposed variational inference technique. We compare the accuracy of point estimates constructed with VIAS and IAS, as well as the uncertainty quantification given by VIAS and the iterative Laplace approximation in Section \ref{sec: hierarchical}. Finally, we show how the ELBO can be used to select the model hyperparameters, and we demonstrate that the accuracy of the reconstruction obtained with this model selection approach is comparable to the accuracy achieved using the true hyperparameters. 

%\begin{figure}
%\centering
%	 \includegraphics[height = 4.75cm]{Fig1.pdf}
%\centering
%\includegraphics[height = 4.75cm, width = 15cm]{VIAS_TRUE_MS_IAS.pdf}
%	\caption{\label{fig:Gen_Data} First row: synthetic truth and data. Second row: computed results with VIAS and IAS. The results for VIAS are more accurate and the $95 \%$ credible intervals are shorter while providing suitable coverage.}
%\end{figure}

\begin{figure}[H]
\centering
\includegraphics[height = 4.5cm]{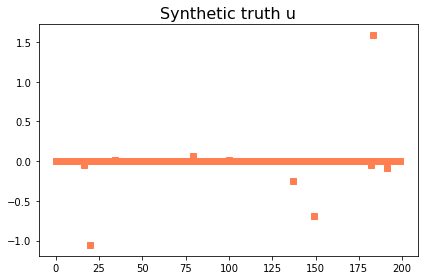}
\includegraphics[height = 4.5cm]{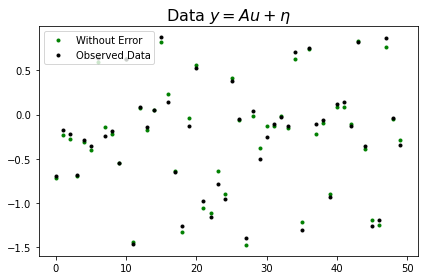}
\includegraphics[height = 4.5cm, width = 5cm]{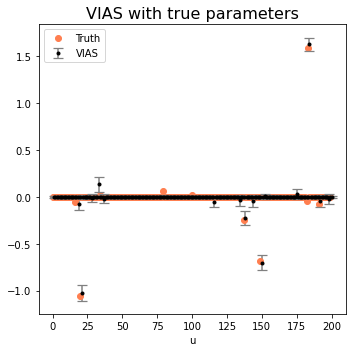}
\includegraphics[height = 4.5cm, width = 5cm]{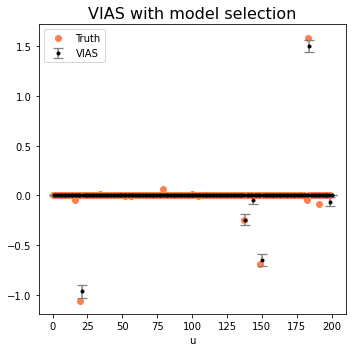}
\includegraphics[height = 4.5cm, width = 5cm]{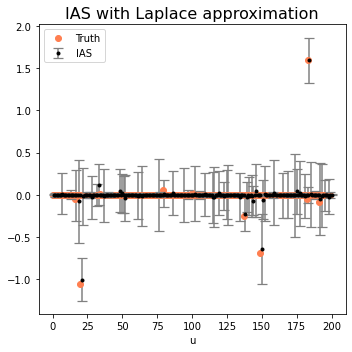}
	\caption{\label{fig:Gen_Data} First row: synthetic truth and data (Subsection \ref{ssec:Hiearchical}). Second row: computed results with VIAS and IAS. The VIAS $95 \%$ credible intervals are shorter while providing suitable coverage.}
\end{figure}

\subsubsection{Setting}
We sample $\theta_i$ values from a gamma distribution with $\alpha = 0.005 \, , \, \beta = 0.05$ (mean $0.1$ and variance $2$). Conditional on these $\theta_i$ values, we generate the synthetic truth $u \in \mathbb{R}^{d}$ by sampling independently $u_i \sim \Nc(0, \bti)$. The data $y \in \mathbb{R}^{n}$ is generated by $y = Au + \eta$,  where  $A \in \mathbb{R}^{n \times d}$ is randomly generated with each entry being uniformly distributed  between $0$ and $1$. 
We choose $ d = 200$ and $n = 50$ so that the problem is severely underdetermined. The error term $\eta$ is sampled from a normal distribution $\Nc(0, \gamma^2 I_{50}),$ where $\gamma$ is chosen to be 5\% of the max-norm of $Au$. Figure \ref{fig:Gen_Data} shows the generated synthetic truth $u$ and data $y$. The generated $u \in \R^{200}$ is very sparse with only $4$ distinctly large components of varying sizes.

\subsubsection{Numerical Results with True Hyperparameters}
Here we report numerical results for VIAS and IAS. For VIAS, we set hyperparameters $\alpha$ and $\beta$ to be exactly those used to generate the data, namely $\alpha = 0.005 \, , \, \beta = 0.05.$ This determines the choice of variational parameters $s = -0.495$ and $b = 0.1$. For IAS, as we expect sparse structure in the parameter of interest on a unit scale, we set the parameter $\alpha = 1$ with $\tilde \beta = 0.00001$. Figure \ref{fig:Gen_Data}  displays the results for VIAS and IAS.  For the initializations, we set both $\theta^{0}$, $m^0 \in \mathbb{R}^{200}$ to be the all-ones vector, and $C^0 \in \mathbb{R}^{200 \times 200}$ to be the identity matrix. Both algorithms yielded successful reconstruction of $u$. The credible intervals for the VIAS estimates are significantly shorter than the intervals from the IAS using Laplace approximation. The implication is that, unlike VIAS, IAS may not give sufficient shrinkage off the support. To quantify the accuracy of these credible intervals, we conduct repeated simulations fixing the matrix $A$ and synthetic truth $u$, while resampling the error term $\eta$ to generate different $y$ values. On conducting 1000 such simulations and generating 200 credible intervals for each component of $u$, we observe that the 95\% credible interval for the components of $u$ covers the true values 96.06\%  of the time for VIAS with the true hyperparameters, and 98.67\% for IAS. Thus, VIAS maintains similar accuracy to  IAS with much narrower credible intervals, providing superior uncertainty estimation. 

The left and the middle plots in Figure \ref{fig:Gen_Convergence} show the convergence of VIAS using the true hyperparameters. The ELBO value stabilizes after 100-200 iterations. We also illustrate the decay of the relative change in max-norm of the variational parameters along the VIAS iterations, which can be seen in the middle plot. The number of iterations that IAS needs to stabilize is significantly lower, of the order of 10. \nc

\subsubsection{Numerical Results with Model Selection}
In this subsection, we investigate the learning of the hyperparameters $\alpha, \beta.$ For this purpose, we use the ELBO as a model selection tool, choosing the hyperparameters which maximize the ELBO. Since the ELBO is a lower bound for the marginal probability of the data, larger ELBO values suggest a better fit to the data.  In practice, we obtained ELBO values after 300 iterations of VIAS for each choice of ($\alpha$, $\beta$) values in a two-dimensional grid. The choice of ($\alpha$, $\beta$) value which led to the maximal ELBO value was used as our hyperparameter values. 

 One would expect the hyperparameters which maximize the ELBO to be close to the true model hyperparameters. However, on conducting the model selection, we find that this is not the case. The hyperparameters found using model selection are  $\alpha = 0.001$ and $\beta = 1623$, whose corresponding ELBO value after 1000 iterations was roughly around $3899$, a significantly larger value than the ELBO with the true hyperparameters, which is $3243$, which can be seen in Figure \ref{fig:Gen_Convergence}. Despite this large difference in the hyperparameters and the ELBO, the resulting reconstructions are similar, as displayed in Figure \ref{fig:Gen_Data}. Relative to the results obtained with the true parameters, the ELBO-selected model slightly underestimates the signal due to overshrinkage induced by using large $\beta$ value, see Figure \ref{fig:Gen_Data}. On conducting repeated simulations in the manner mentioned previously, the 95\% credible intervals for this ELBO-selected model contain the true values around 91-92\% of the time. The model selected by maximizing the ELBO provides credible intervals with good accuracy. If no prior knowledge on hyperparemeters $\alpha$ and $\beta$ is available, we propose calibrating hyperparameters based on the ELBO as a way to suitably balance between shrinkage and estimation of parameters.

\begin{figure}
\centering
		 \includegraphics[height = 4.5cm, width = 0.3 \linewidth]{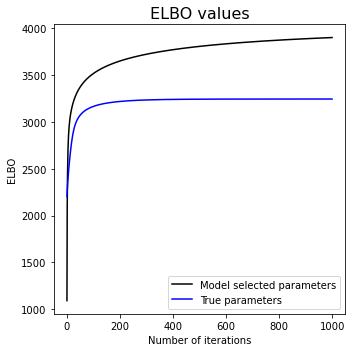}
		 \includegraphics[height = 4.5cm, width = 0.3 \linewidth]{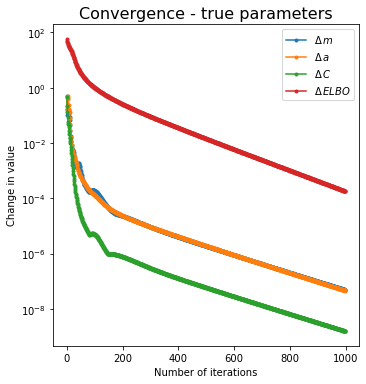}
		 \includegraphics[height = 4.5cm, width = 0.3 \linewidth]{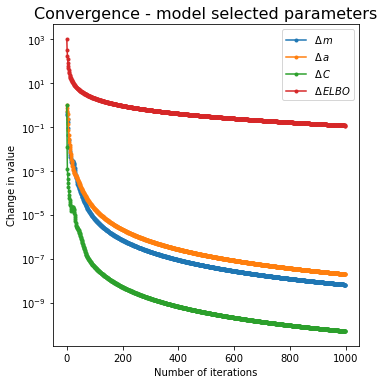}
		\caption{\label{fig:Gen_Convergence}  Convergence of the ELBO and the variational parameters along VIAS iterates, with truth and data generated from the hierarchical model (Subsection \ref{ssec:Hiearchical}).}
\end{figure}

\subsection{Fixed Sparse Truth}\label{ssec:fixedsparsetruth}
In this second example, we consider a fixed truth which is less sparse than the one in the previous example.
Moreover, the truth is chosen rather than sampled generatively. The model hyperparameters $\alpha$ and $\beta$ are chosen according to the model selection procedure to maximize the ELBO, which gave $\alpha = 0.0001$ and $\beta = 33.59$. 
\subsubsection{Setting}
We generate a random matrix $A \in \mathbb{R}^{50 \times 100}$ with entries taking values between 0 and 1. The to-be-reconstructed parameter $u \in \R^{100}$ is chosen so that only $10$ components are non-zero, see Figure \ref{fig:Example2}. The data $y$ is generated by multiplying $A$ with $u$ and adding a randomly sampled Gaussian with  standard deviation taken as 2\% of the max-norm of $Au$. As in the previous example, we set both $\theta^{0}$, $m^0 \in \mathbb{R}^{100}$ to be the all-ones vector, and $C^0\in \mathbb{R}^{100 \times 100}$ to be the identity matrix.

\subsubsection{Numerical Results}
Figure \ref{fig:Example2} shows the VIAS results and a comparison to other techniques. The VIAS predictions are close to the true values and even when the prediction was not accurate, credible intervals successfully captured true values. Compared to IAS, VIAS point estimates are much closer to the true values. The IAS reconstruction is less sparse than the VIAS reconstruction and typically underestimates the non-zero components of the signal. We have also obtained LASSO estimates with tuning parameter calibration based on cross validation (CV), Akaike information criterion, and Bayesian information criterion using Python's \texttt{LassoLarsIC} and \texttt{LassoLarsCV} functions. We only report in Figure \ref{fig:Example2} the result based on CV, which was the most accurate. From Figure \ref{fig:Example2}, we can observe that the VIAS estimate was superior to IAS and LASSO in terms of estimating zero components of the parameter, while also maintaining a good accuracy in non-zero components.

\begin{figure} 
    \centering
    \includegraphics[height = 4.5cm, width = 7.5cm]{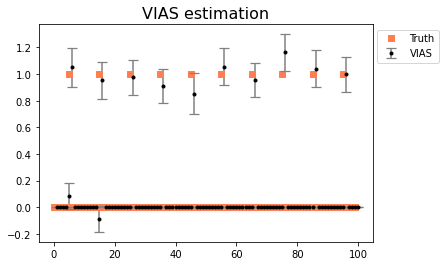}
    \includegraphics[height = 4.5cm, width = 7.5cm]{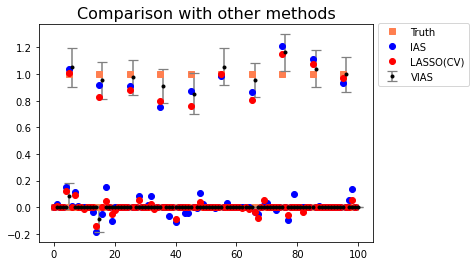}
    \caption{\label{fig:Example2}  Fixed sparse truth (Subsection \ref{ssec:fixedsparsetruth}). VIAS reconstruction and credible intervals (left). Comparison with IAS and LASSO with cross validation (right).}	
\end{figure}

\subsection{Deconvolution}\label{ssec:deconvolution}

\begin{figure}
	    \centering
		\includegraphics[height = 4.5cm, width = 18cm]{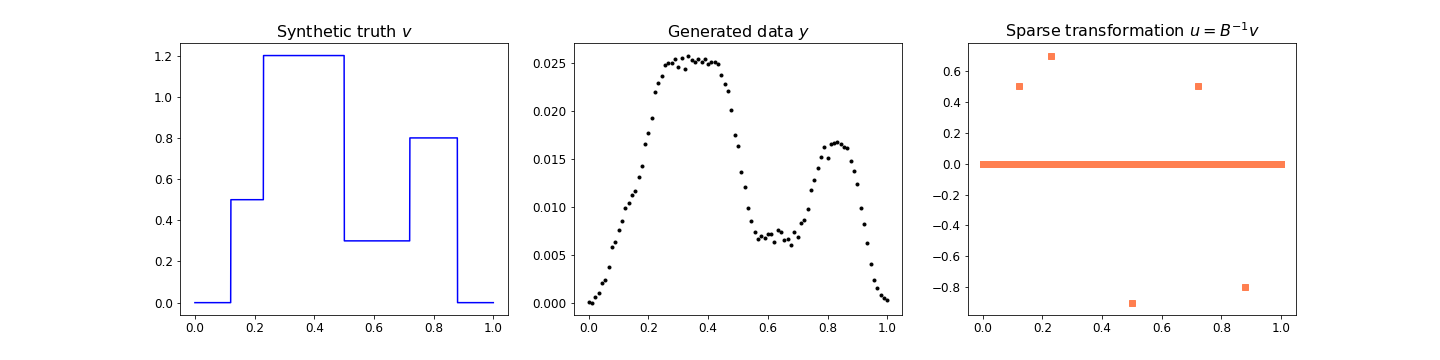}	\includegraphics[height = 4.5cm, width = 18cm]{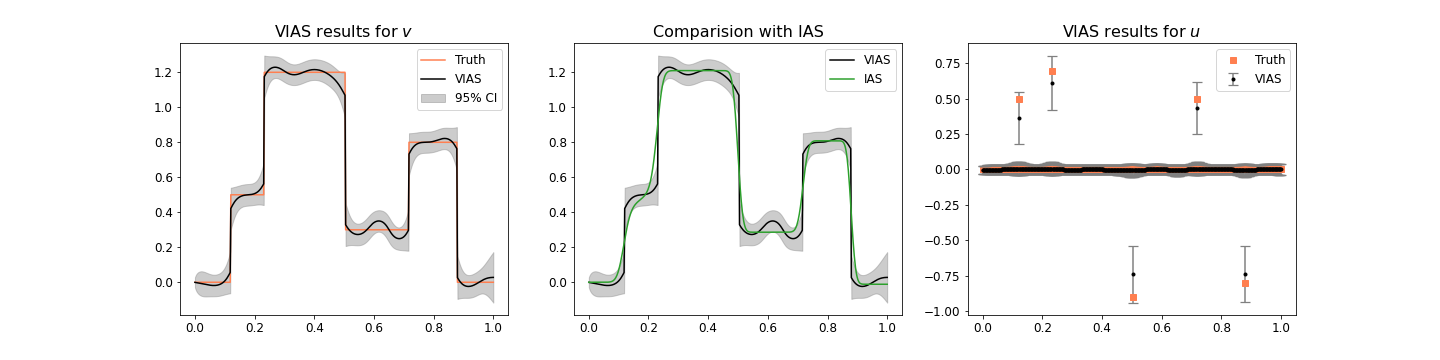}
		\includegraphics[height = 4.5cm, width = 18cm]{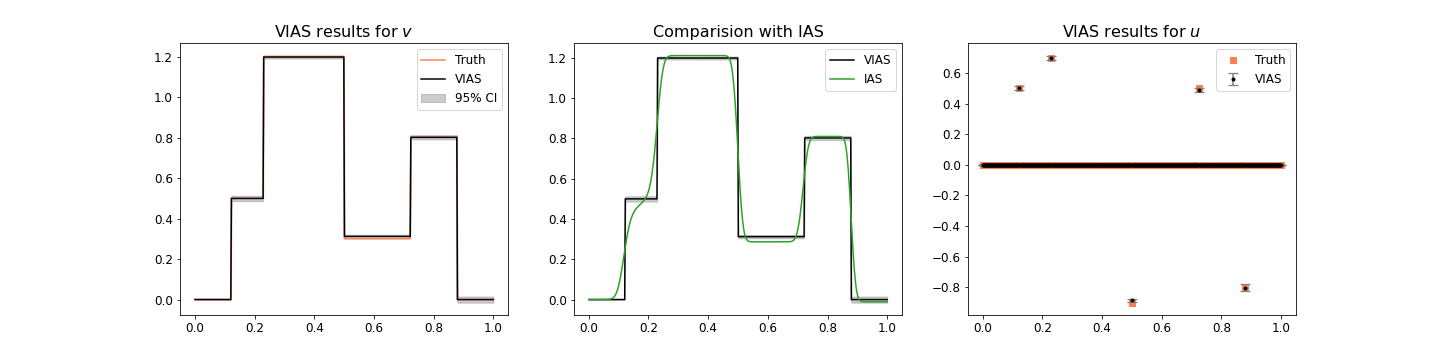}
		%}
		\caption{\label{fig:Ex3_Data} Deconvolution problem (Subsection \ref{ssec:deconvolution}). First row: truth, data, and sparse representation. Second row: VIAS reconstruction of the signal and its sparse representation with user-chosen model hyperparameters $\alpha = 0.12$ and $\beta = 50$. Third row: same as second row, but with ELBO-selected model.}
\end{figure}

In this example, we consider the 1D deconvolution problem in \cite{calvetti2020sparse}, where the goal is to reconstruct a piecewise constant signal convolved with an Airy kernel. We compare the results obtained with IAS and VIAS.  We demonstrate the high accuracy achieved by VIAS with ELBO-selected model hyperparameters and show that the VIAS signal covariance provides meaningful uncertainty quantification on the reconstruction. 

\subsubsection{Setting}
% New Figure

Let $f:[0,1]\to \R$ be a piecewise constant function with  $f(0)=0$. The data $y$ is generated by the following convolution:
	\[
	y_j = \int_0^1 A(s_j - t) f(t) dt + \eta_j,\quad 1\leq j\leq n,  \quad A(t) = \left(\frac{J_1(\kappa|t|)}{\kappa|t|}\right)^2, 
	\]
	where $J_1$ is the Bessel function of the first kind, $\kappa$ is a scalar controlling the width of the kernel that we set to $\kappa = 40$, and $s_j = (4+j)/100$. The above integral can be discretized, leading to the linear equation
\begin{equation}
\label{eq:Model_Example3}
	 y = Av + \eta \, ,  \quad A_{jk} =  w_k A(s_j - t_k) \,   ,\quad \eta \sim \Nc(0, \gamma^2 I_{n}) \, , 
\end{equation}
where $v \in \R^d$ has components $v_k = f(t_k)$ with $t_k = (k-1)/(n-1)$, and the $w_k$ are quadrature weights for discretization of the integral. The standard deviation $\gamma$ is set to be $1 \%$ of the max-norm of the noiseless signal. 
	
The unknown parameter $v$ is not a sparse vector, but can be written in sparse form in a suitable basis. To that end, define $u_j = v_j - v_{j-1}$ with $u_0 = 0$. Since $v$ is piecewise constant, $u$ is sparse. Note that we can write $u = B^{-1}v,$ where 
$$
	B^{-1}=\left[\begin{array}{cccc}
	1&0&\ldots&0\\
	-1&1&\ldots&0\\
	&&\ddots&\\
	0&\ldots&-1&1\\
	\end{array}\right]
	\in \R^{d \times d }.
$$
Thus, we can rewrite (\ref{eq:Model_Example3}) in terms of this sparse unknown vector $u$ as follows:
	\begin{equation}
	y = ABu + \eta, \quad \eta \sim \Nc(0, \gamma^2 I_{n}). 
	\end{equation}
	
 Our inverse problem is to estimate the vector $u$, assumed to be sparse, from the data vector $y$. Figure \ref{fig:Ex3_Data} shows the piecewise constant function $v$ to be reconstructed, its sparse transformation $u$, and the data $y$. We take $d = 500$ and $n = 91$ so that the problem is underdetermined. 
	
\subsubsection{Reconstruction Accuracy and Model Selection}

% New Figure
%\begin{figure}
%\centering
%		\centerline{ \includegraphics[height = 5cm, width = 0.88 \linewidth]{Figures/Example3/Fig2.pdf}	}
%		\caption{\label{fig:Ex3_Results} Example 3 results }
%\end{figure}

The results of applying VIAS to this problem are displayed in Figure \ref{fig:Ex3_Data}. As in the previous example, we set $m^0 \in \mathbb{R}^{500}$ to be the all-ones vector, and $C^0\in \mathbb{R}^{500 \times 500}$ to be the identity matrix. We consider two implementations of VIAS. In the first (Figure \ref{fig:Ex3_Data}, middle row), we use hyperparameter values $\alpha = 0.12$ and $\beta = 50$. In the second (Figure \ref{fig:Ex3_Data}, bottom row), we adopt the ELBO approach for model selection, which gives hyperparameter values $\alpha = 0.0001$ and $\beta \approx 7742$. Both VIAS implementations produce sparse solutions and the true signal lies within the obtained credible intervals.
%These results in the sparse space can be transformed back into the non-sparse signal space. 
The signal reconstructed with VIAS has sharp jumps whereas the IAS reconstruction has much smoother jumps. For instance, the first two jumps are treated as a smooth transition by IAS, while VIAS successfully detects them as separate jumps. While VIAS with hyperparameters $\alpha = 0.12$ and $\beta = 50$ detects the presence of five distinct jumps, uncertainty remains in the location of the jumps. Moreover, the reconstruction exhibits oscillatory artifacts in the constant regions, where IAS remains accurate. The VIAS reconstruction with model selection is significantly more accurate; it detects the five jump locations and it does not show oscillatory artifacts.

\subsubsection{Uncertainty Quantification and Covariance Structure}

The credible intervals obtained with both VIAS implementations provide additional insight into the nature of uncertainty quantification in this problem. VIAS detects that the main source of uncertainty is the location of the jumps. We observe in Figure \ref{fig:Ex3_Data} that the uncertainty in the reconstruction spikes near the jumps, while it remains relatively small around the constant regions. Moreover, the enhanced accuracy of VIAS with model selection is accompanied by narrower confidence intervals. Notice that  we have imposed the condition $v_0 = 0$, hence we are certain that the value at $0$ is $0$, and uncertainty is expected to increase from left to right. This overall trend is also successfully identified by both VIAS implementations. 

\vspace{-0.6cm}
		\begin{minipage}{0.45\textwidth}
		\begin{figure}[H]
		\captionsetup{width=.8\linewidth}
\hspace{-0.5cm}	\includegraphics[height = 4.5cm, width = 8cm]{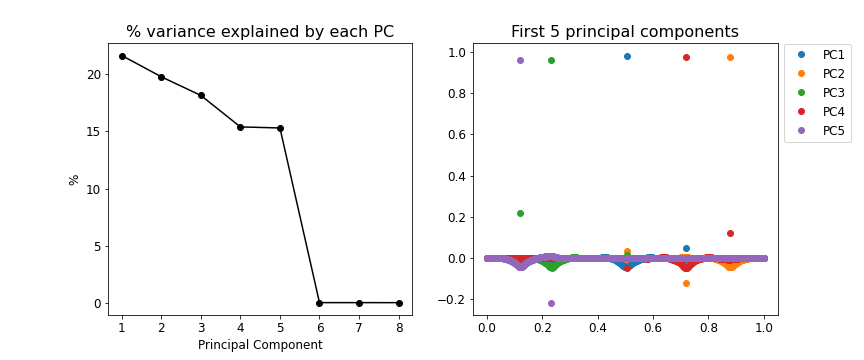}	
		\caption{\label{fig8} Principal component analysis of VIAS uncertainty with hyperparameters $\alpha = 0.12, \beta = 50.$}
		\end{figure}
		\end{minipage}
		\begin{minipage}{0.45\textwidth}
    	\begin{figure}[H]
		\captionsetup{width=.8\linewidth}
\hspace{-0.7cm}	\includegraphics[height = 4.5cm, width = 8cm]{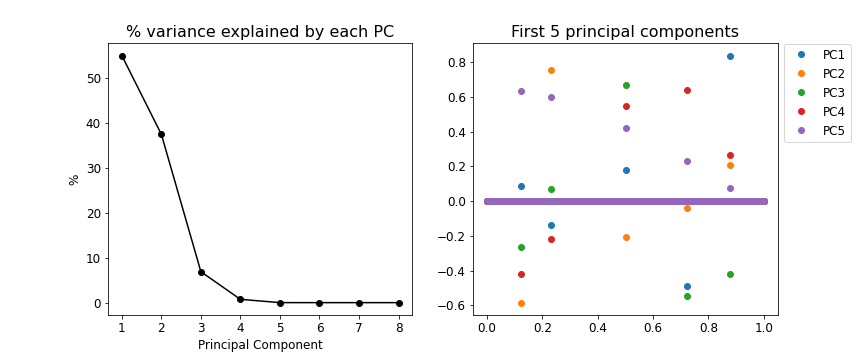}
		\caption{\label{fig9} Principal component analysis of VIAS uncertainty with ELBO-selected hyperparameters.}
	\end{figure}
		\end{minipage}
\\

An additional benefit of VIAS is that it not only gives an approximation to the component-wise variances on the signal reconstruction, but also an approximate covariance matrix. This matrix contains information on the dependencies in the reconstruction of various components of the signal. To illustrate this point, Figures \ref{fig8} and \ref{fig9} show a principal component analysis of the covariance matrix $C$ obtained with the two VIAS implementations considered above. For VIAS with user-chosen hyperparameters $\alpha = 0.12$ and $\beta = 50,$  the first five principal components each explain 15-20\% of the variance, and this drops to less than 0.1\% after the 5th component. %Hence, the first five principal components together explain around 90\% of the variance in the signal reconstruction. 
In contrast, the three principal components already explain most of the variance for VIAS with model selection. In both VIAS implementations, each principal component is localized around a jump in the signal. With model selection, the principal components are fully localized at the jumps, reflecting that no uncertainty remains in the location of the jumps. On the other hand, with the first VIAS implementation the localization of principal components at the jumps is not perfect, reflecting that there is non-negligible uncertainty in both the location and magnitude of the jumps. 
Each principal component also contains components of its nearby jumps, which suggests that nearby jumps may be correlated. For instance, this is noticeable in the first VIAS implementation PC3 and PC5 ---corresponding to the first and second jumps--- likely because the first two jumps are very close to each other and it is harder to distinguish the two. Overall, the principal components obtained with VIAS successfully identify that most of the variance in the signal reconstruction lies around the jumps, that nearby jumps are correlated, and that uncertainty in the problem propagates from left to right. Finally, this example demonstrates the effect of the hyperparameters in the reconstruction. Poor choice of the hyperparameters gives a less accurate, more uncertain reconstruction.

\subsection{Learning Dynamics of Lorenz-63 System}\label{ssec:LearningDynamics}
In this section, we illustrate the use of IAS and VIAS for sparse identification of dynamical systems. Our problem setting is motivated by \cite{brunton2016discovering}.
\subsubsection{Setting}
Consider the Lorenz-63 system  \cite{lorenz1963deterministic}
\begin{align}\label{eq:lorenz63}
\begin{split}
    \frac{dx}{dt} &= \sigma(y-x), \\
    \frac{dy}{dt} &= x(\rho-z) - y, \\
    \frac{dz}{dt} &= xy-\zeta z, 
\end{split}    
\end{align}
with the classical parameter values  $\sigma = 10, \rho = 28$, $\zeta =8/3$. Our goal is to recover the right-hand side of \eqref{eq:lorenz63} from time series data.  We assume to have observations of a trajectory and its derivative along $2000$ equidistant time points in the time-interval $[0, 40];$  thus the time between observations is $\Delta t = 0.02.$ Note that in this example $y$ denotes the second component of the dynamics rather than the observed data.

Following \cite{brunton2016discovering}, we adopt a dictionary-learning strategy and construct, from the given trajectory data, a matrix of the form: 
\[
A = 
\left[
  \begin{array}{cccccccccccc}
    \vert & \vert & \vert & \vert & \vert & \vert  & \vert & \vert  & \vert & & \vert \\
    \textbf{x} & \textbf{y} & \textbf{z} & \textbf{x} ^2 & \textbf{y}^2 & \textbf{z}^2 & \textbf{x} \textbf{y} & \textbf{x} \textbf{z} & \textbf{y} \textbf{z} & \ldots & \textbf{x} ^5 & \ldots    \\
   \vert & \vert & \vert & \vert & \vert & \vert  & \vert & \vert  & \vert & & \vert 
  \end{array}
\right] \in \mathbb{R}^{2000 \times 55}.
\]
We then obtain, as in \cite{brunton2016discovering}, synthetic data on the derivatives by setting 
\begin{align*}
    \dot{\textbf{x}} &= A\Phi_1 + \eta_1, \quad  \Phi_1 = [-10, 10, 0, 0, 0, 0, 0, \ldots, 0]^\top,  \\
    \dot{\textbf{y}} &= A\Phi_2 + \eta_2,   \quad \Phi_2 = [28, -1, 0, 0, -1, 0, 0, \ldots, 0]^\top,\\
    \dot{\textbf{z}} &= A\Phi_3 + \eta_3, \quad \Phi_2 = [0, 0, -8/3, 0, 0, 0, 1, \ldots, 0]^\top,
\end{align*}
where $\eta_i \sim \mathcal{N}(0,0.3I_{2000})$ are independent. Our goal is to recover $\Phi_1, \Phi_2, \Phi_3$ based on $A$, $\dot{\textbf{x}}$, $\dot{\textbf{y}}$, $\dot{\textbf{z}}$. Note that $\Phi_1, \Phi_2, \Phi_3$ are sparse. More generally, the data-driven learning of dynamical systems in \cite{brunton2016discovering} relies on the underlying assumption that only a few terms of a given dictionary (in our example made of polynomials of degree five) govern the dynamics; sparse-promoting VIAS is hence a natural algorithm for identification of dynamical systems. 

\subsubsection{Numerical Results}
The recovery of the Lorenz-63 model via VIAS is shown in Figure \ref{Figure6}. Since we expect sparse structure in the parameter of interest, same as in our first example, we use $s = -0.495$ and $b = 0.1$ for the variational parameters. For the initializations, we set $\theta^0, m^0 \in \mathbb{R}^{55}$ to be the all-ones vector, and $C^0 \in \mathbb{R}^{55 \times 55}$ to be the identity matrix.
We observe that VIAS accurately recovers the true parameter values. As VIAS can quantify uncertainties of our estimates, we provide the true dynamics in the blue line in Figure \ref{Figure9} with the shaded regions determined by dynamics obtained from 2.5 and 97.5 percentile credible levels of the true parameters. We point out that despite the chaotic behavior of the Lorenz-63 system, the uncertainty in the dynamics remains moderate due to the high accuracy of the recovered coefficients. Moreover, we note that the relative larger error in the coefficients of the $y$-trajectory in Figure \ref{Figure6} translates into wider credible intervals for the reconstructed trajectories of $y$ in Figure \ref{Figure9}. Therefore, VIAS correctly identifies that there is more uncertainty in the reconstruction of the $y$-component.

Compared to VIAS, IAS showed inferior performance in estimating parameters of the Lorenz-63 model as one can see in Figure \ref{Figure6}. Furthermore, we provide plots of true dynamics with shaded regions determined by dynamics recovered from 95 percent credible intervals obtained from a Laplace approximation to the posterior. As shown in Section \ref{ssec:Hiearchical}, credible intervals based on Laplace approximation tend to be larger than the ones obtained from VIAS. In the context of the Lorenz-63 model, the uncertainty in the dynamics is amplified by the mismatch between the estimated coefficients and the true coefficients. From Figure \ref{Figure9}, we can see that quantifying uncertainty of dynamics based on Laplace approximation gives little information as the constructed shaded regions are often too wide, which highlights the strength of VIAS in uncertainty quantification tasks.

%\begin{figure}
%		\centering \includegraphics[height = 5.5cm, width = 5cm]{VIAS_LORENZ_X_TRAJ_3.png}
%\includegraphics[height = 5.5cm, width = 5cm]{VIAS_LORENZ_X_TRAJ_6.png}
%\includegraphics[height = 5.5cm, width = 5cm]{VIAS_LORENZ_X_TRAJ_9.png}	\caption{\label{Figure6} Recovery of parameters for x-trajectory of Lorenz-63 model using VIAS. Left: 3 VIAS iterations, Middle: 6 VIAS iterations, Right: 9 VIAS iterations}
%\end{figure}
%
%\begin{figure}
%		\centering \includegraphics[height = 5.5cm, width = 5cm]{VIAS_LORENZ_Y_TRAJ_3.png}
%\includegraphics[height = 5.5cm, width = 5cm]{VIAS_LORENZ_Y_TRAJ_6.png}
%\includegraphics[height = 5.5cm, width = 5cm]{VIAS_LORENZ_Y_TRAJ_9.png}	\caption{\label{Figure7} Recovery of parameters for y-trajectory of Lorenz-63 model using VIAS. Left: 3 VIAS iterations, Middle: 6 VIAS iterations, Right: 9 VIAS iterations}
%\end{figure}
%
%\begin{figure}
%		\centering \includegraphics[height = 5.5cm, width = 5cm]{VIAS_LORENZ_Z_TRAJ_3.png}
%\includegraphics[height = 5.5cm, width = 5cm]{VIAS_LORENZ_Z_TRAJ_6.png}
%\includegraphics[height = 5.5cm, width = 5cm]{VIAS_LORENZ_Z_TRAJ_9.png}	\caption{\label{Figure8} Recovery of parameters for z-trajectory of Lorenz-63 model using VIAS. Left: 3 VIAS iterations, Middle: 6 VIAS iterations, Right: 9 VIAS iterations}
%\end{figure}

\begin{figure}[H]
		\centering 
	\includegraphics[height = 4.3cm]{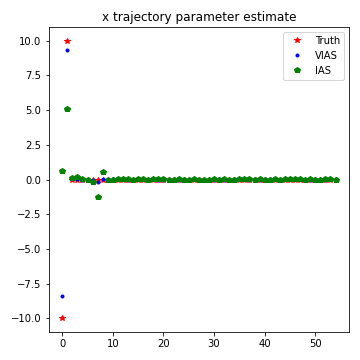}
	\includegraphics[height = 4.3cm]{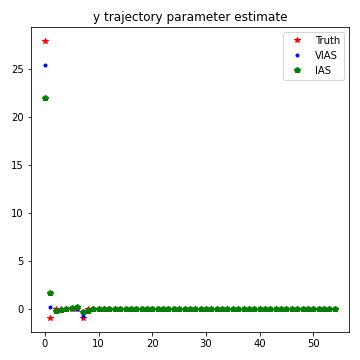}
	\includegraphics[height = 4.3cm]{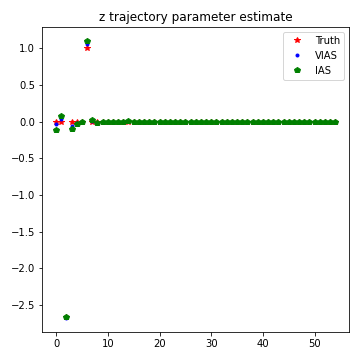}
	\includegraphics[height = 4.3cm]{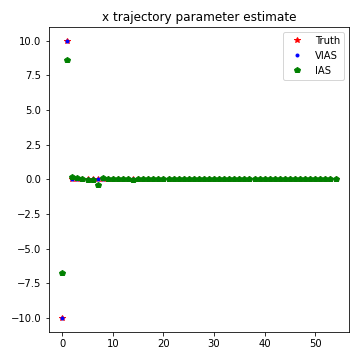}
	\includegraphics[height = 4.3cm]{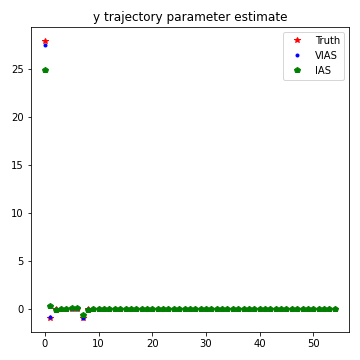}
	\includegraphics[height = 4.3cm]{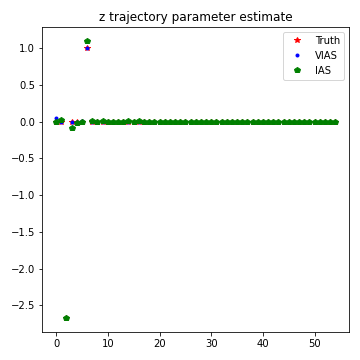}
	\caption{\label{Figure6} Recovery of dictionary coefficients for x-trajectory (first column), y-trajectory (second column), and z-trajectory (third column) using IAS and VIAS. Top: two iterations. Bottom: five iterations.}
\end{figure}

\begin{figure}[H]
		\centering 
	\includegraphics[height = 4.5cm]{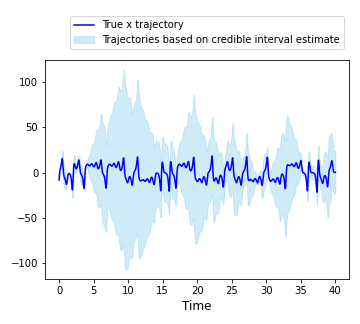}
	\includegraphics[height = 4.5cm]{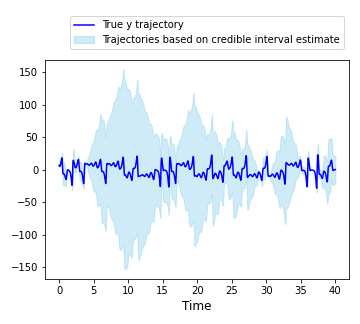}
	\includegraphics[height = 4.5cm]{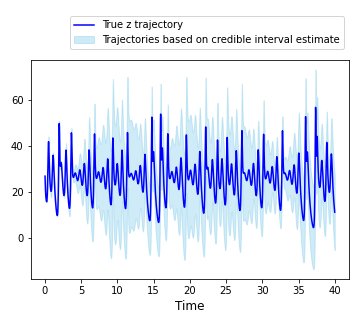}
	\includegraphics[height = 4.5cm]{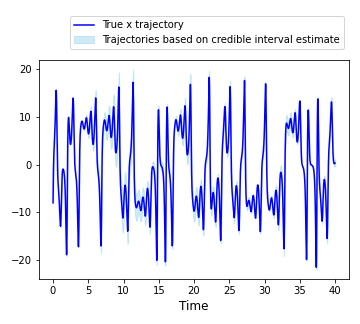}
	\includegraphics[height = 4.5cm]{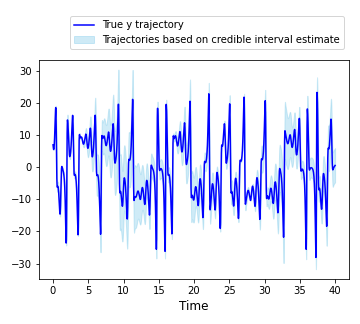}
	\includegraphics[height = 4.5cm]{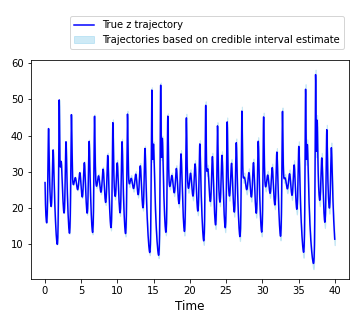}
	\caption{\label{Figure9} True Lorenz-63 trajectory and VIAS estimation. Top: two VIAS iterations. Bottom: five VIAS iterations. Blue line is the true dynamics. Shaded regions are constructed from 2.5 and 97.5 credible levels of coefficients.}
\end{figure}

\begin{figure}
		\centering 
	\includegraphics[height = 4.5cm]{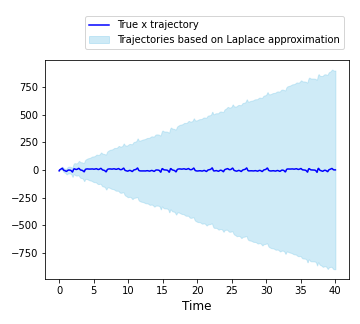}
	\includegraphics[height = 4.5cm]{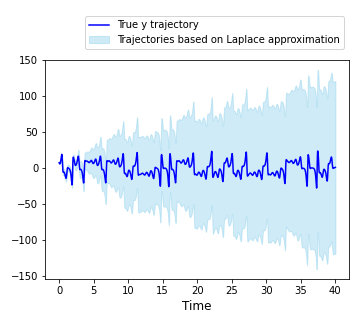}
	\includegraphics[height = 4.5cm]{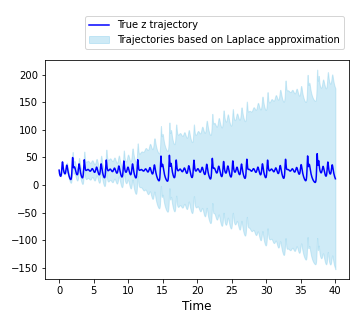}
	\includegraphics[height = 4.5cm]{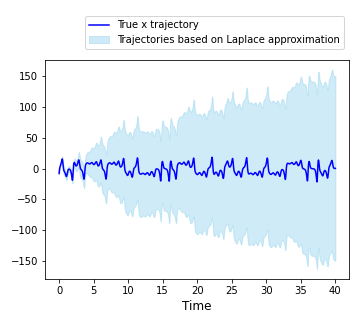}
	\includegraphics[height = 4.5cm]{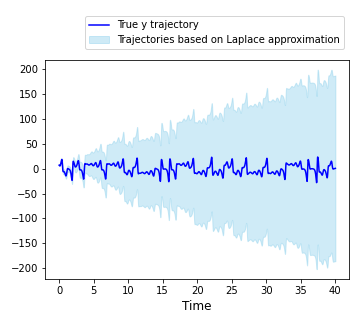}
	\includegraphics[height = 4.5cm]{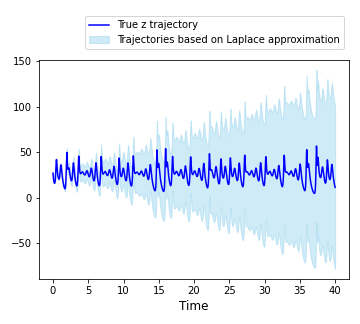}
	\caption{\label{Figure11} True Lorenz-63 trajectory and IAS estimation. Top: two IAS iterations. Bottom: five IAS iterations. Blue line is the true dynamics. Shaded regions are constructed from 2.5 and 97.5 credible levels of coefficients based on Laplace approximation.}
\end{figure}

%\begin{figure}
%		\centering 
	%\includegraphics[height = 5.5cm, width = 5cm]{VIAS_LORENZ_X_CI.png}
%	\includegraphics[width = 5cm]{LORENZ_VIAS_TRAJEC_X.png}
	%\includegraphics[height = 5.5cm, width = 5cm]{VIAS_LORENZ_Y_CI.png}
%		\includegraphics[width = 5cm]{LORENZ_VIAS_TRAJEC_Y.png}
	%\includegraphics[height = 5.5cm, width = 5cm]{VIAS_LORENZ_Z_CI.png}	
%		\includegraphics[width = 5cm]{LORENZ_VIAS_TRAJEC_Z.png}
%	\caption{\label{Figure9} True Lorenz-63 trajectory and VIAS estimation. Blue line: true dynamics. Shaded regions are constructed from 2.5 and 97.5 credible levels of coefficients.}
%\end{figure}

\section{Conclusion and Future Directions}\label{sec:Conclusions}
This paper introduced VIAS, a variational inference computational framework for linear inverse problems with gamma hyperpriors. The proposed VIAS shares the flexibility and ease of implementation of IAS for MAP estimation. We have shown the accuracy of VIAS in several computed examples, and we have explored its potential to provide meaningful uncertainty quantification and perform model selection. There are several research directions that stem from this work:

\begin{itemize}
\item We have established a local convergence result for VIAS, but we have not provided an analysis of convergence rates. Moreover, it would be interesting to study the approximation error between the variational distribution and the true posterior.
\item Combining VIAS with iterative ensemble Kalman methods \cite{chada2020iterative} may allow to extend the current variational framework to nonlinear inverse problems, and to enhance the scalability to high dimensional linear and nonlinear inverse problems. In addition, we also envision that VIAS may provide a natural way to promote sparsity in iterative ensemble Kalman methods that are based on $L^2$ penalties. 
\item We have explored the potential of VIAS to perform approximate Bayesian inference and provide meaningful uncertainty quantification. In future work, our variational approach will be combined with Markov chain Monte Carlo \cite{de2013variational} and sequential Monte Carlo \cite{naesseth2018variational} for fully-Bayesian inference.
\item More general hyperpriors could be considered within our variational framework. In this direction, the work \cite{calvetti2020sparse} has investigated more flexible generalized gamma hyperpriors in the context of MAP estimation. 
\end{itemize}

\section*{Acknowledgments}
DSA is thankful for the support of NSF and NGA through the grant DMS-2027056 and to the BBVA Foundation for the Jos\'e Luis Rubio de Francia start-up grant. 
The work of HK was partially supported by the grant DMS-2027056.

\bibliographystyle{unsrt} 
\bibliography{references}
\end{document}